\RequirePackage{fix-cm}
\documentclass[smallextended]{svjour3}       
\smartqed  
\usepackage{graphicx}
\usepackage[utf8]{inputenc}
\usepackage{amsmath,amssymb,color,graphicx}
\usepackage[usenames,dvipsnames,svgnames,table]{xcolor}
\usepackage{blkarray}
\usepackage{verbatim}
\usepackage{rotating}
\usepackage{tcolorbox}
\usepackage{algorithm,verbatim}
\usepackage[noend]{algpseudocode}
\usepackage{hyperref}
\usepackage{xspace}
\usepackage{array}
\usepackage[disable]{todonotes}
\usepackage{wrapfig}

\begin{document}

\title{Emanation Graph: A Plane Geometric Spanner with Steiner Points\thanks{A preliminary version of this work was presented at  the 30th Canadian Conference on Computational Geometry (CCCG)~\cite{DBLP:conf/cccg/HamedmohseniRM18} and the 46th International Conference on Current Trends in Theory and  Practice of Computer Science (SOFSEM)~\cite{DBLP:conf/sofsem/HamedmohseniRM20}.}
}


\author{Bardia Hamedmohseni        \and
        Zahed Rahmati \and
        Debajyoti Mondal
}


\institute{B. Hamedmohseni  \at
              Department of Mathematics and Computer Science, Amirkabir University of Technology \\
              \email{fhamedmohseni@aut.ac.ir}           
           \and
           Z. Rahmati\at
              Department of Mathematics and Computer Science, Amirkabir University of Technology \\
              \email{zrahmati@aut.ac.ir}  
              \and
              D. Mondal\at
              Department of Computer Science, University of Saskatchewan, Canada\\
              \email{dmondal@cs.usask.ca}
}

\date{Received: date / Accepted: date}

\maketitle
 
\begin{abstract}
An emanation graph of grade $k$ on a set of points is a plane spanner made by shooting $2^{k+1}$ equally spaced rays from each point, where the shorter rays stop the longer ones upon collision. The collision points are the Steiner points of the spanner. Emanation graphs of grade one were studied by Mondal and Nachmanson 
 in the context of network visualization. They proved that the spanning ratio of such a graph is bounded by $(2+\sqrt{2})\approx 3.414$. We improve this upper bound to $\sqrt{10} \approx 3.162$ and show this to be tight,  i.e., there exist emanation graphs with spanning ratio $\sqrt{10}$. We show that for every fixed $k$, the emanation graphs of grade $k$ are constant spanners, where the constant factor depends on $k$. 
 An emanation graph of grade two   may have twice the number of edges compared to grade one graphs. Hence we
 introduce a heuristic method for simplifying  them.
 In particular,  we compare simplified emanation graphs against Shewchuk's constrained Delaunay triangulations 
on both synthetic and real-life datasets. Our experimental results reveal that the simplified emanation graphs outperform constrained Delaunay triangulations in  common quality measures (e.g., edge count, angular resolution, average degree, total edge length) while maintaining a comparable spanning ratio and Steiner point count.
\end{abstract}

\section{Introduction}
Let $G$ be a geometric graph embedded in the plane, weighted with Euclidean distances. Let $u$ and $v$ be a pair of vertices in $G$. Let $d_G(u,v)$ and $d_E(u,v)$ be the minimum graph distance (i.e., shortest path distance) in $G$ and Euclidean distance between $u$ and $v$, respectively. The \textit{spanning ratio} of  $G$ is $\max\limits_{\{u,v\}\in G} \frac{d_G(u,v)}{d_E(u,v)}$, i.e., the maximum ratio between  $d_G(u,v)$ and $d_E(u,v)$  over all pairs of vertices $\{u,v\}$ in $G$. The graph $G$ is called a \emph{$t$-spanner} of the complete geometric graph, if for every pair of vertices $\{u,v\}$ in $G$, the  distance $d_G(u,v)$ is at most $t$ times $d_E(u,v)$.

We examine \emph{plane geometric spanners}~\cite{DBLP:journals/comgeo/BoseS13,DBLP:conf/imr/Owen98}, i.e., no two edges in the spanner cross except at their common endpoints. A natural question in this context is as  follows: Given a set of points $P$ of $n$ points in the plane, can we compute a planar spanner $G=(V,E)$ of $P$ with small size, degree and spanning ratio? We allow the spanner to have Steiner points, i.e., $P\subseteq V$, thus $V$ may contain vertices that do not correspond to any point of $P$. We do not require the paths between a pair of Steiner points nor between a point of $P$ and a Steiner point to have bounded spanning ratio. Thus the  \textit{spanning ratio of a graph $G$ with Steiner points} is $\max\limits_{\{u,v\}\in P} \frac{d_G(u,v)}{d_E(u,v)}$.

Note that keeping the degree, size and spanning ratio of the spanners small is often motivated by application areas, and appeared in the literature~\cite{DBLP:journals/algorithmica/BoseHS18,DBLP:journals/comgeo/BoseS13}.  Nachmanson \textit{et al.}~\cite{DBLP:conf/gd/NachmansonPLRHC15} introduced a system called GraphMaps for interactive visualization of large graphs based on constrained Delaunay triangulations. Later, Mondal and Nachmanson~\cite{mondal2017}  introduced and used a specific mesh  called  the competition mesh to improve GraphMaps (Figure~\ref{abs}). Given a set of points $P$, a \emph{competition mesh} is constructed by shooting from each point, four axis-aligned rays at the same speed, where the shorter rays stop the longer ones upon collision (the rays that are not stopped are clipped by the axis-aligned bounding box of $P$).  This can also be seen as a variation of a motorcycle graph~\cite{DBLP:journals/dcg/EppsteinE99}, which is constructed by the tracks of $n$ motorcycles as follows: All motorcycles start from their initial positions with fixed velocities assigned to them. If a motorcycle meets the track left by another motorcycle, then it crashes or stops.  If two motorcycles collide, both of them crash simultaneously. Note that this is different in a competition mesh, where one of two motorcycles (or rays) is stopped arbitrarily. Motorcycle graphs have been used to solve various computational geometry problems such as in mesh partitioning~\cite{EppsteinGKT08} and computing the straight skeleton of a polygon~\cite{DBLP:journals/talg/ChengMV16}.

\begin{figure*}[pt]   
\centering
\smallskip 
{\includegraphics[width=.3\textwidth, height = 3cm]{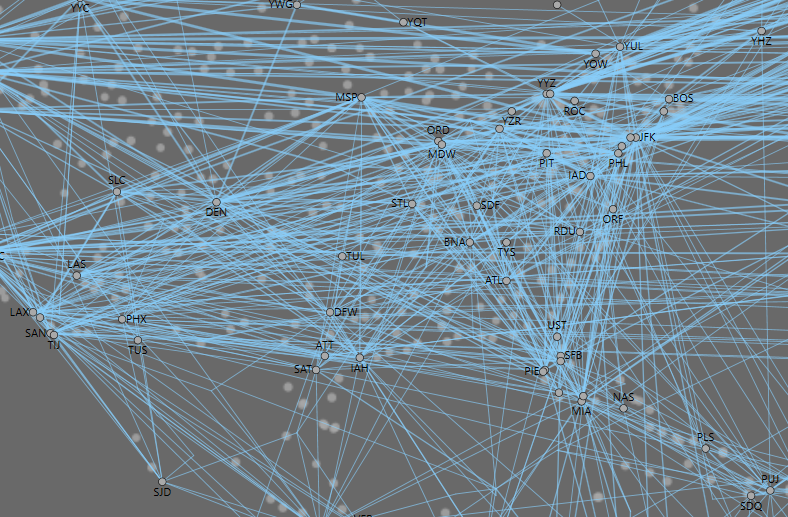}}
\hfill
{\includegraphics[width=.3\textwidth, height = 3cm]{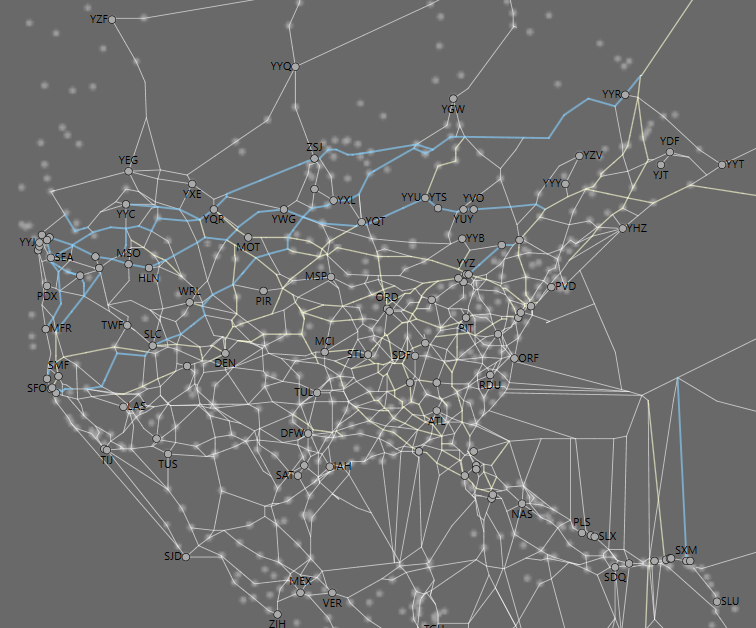}}
\hfill
{\includegraphics[width=.3\textwidth, height = 3cm]{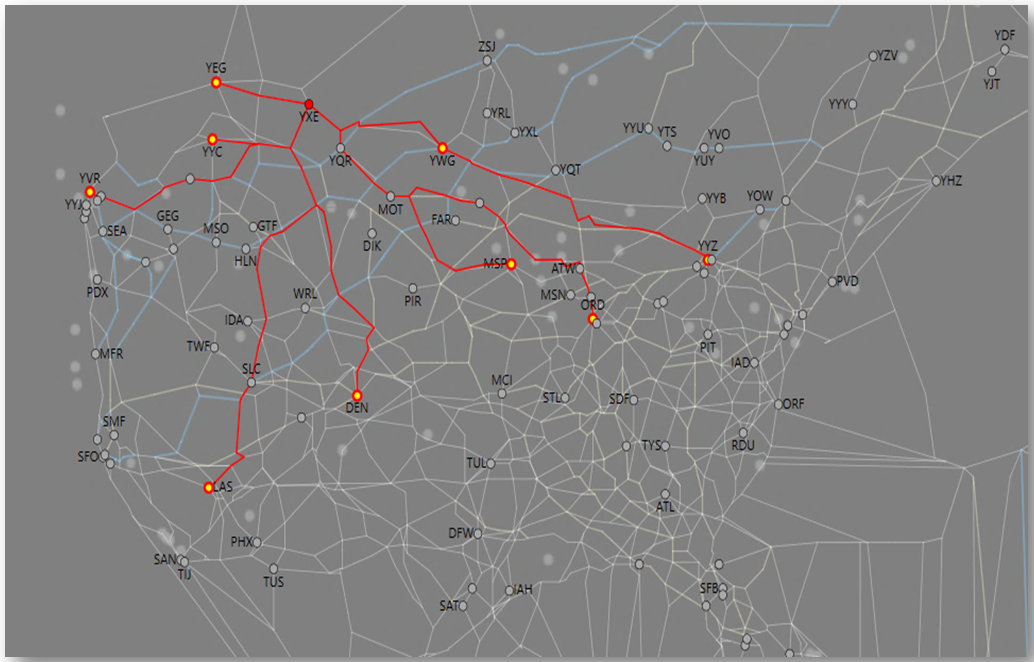}}
\caption{(left) A partial node-link diagram of a flight network. (middle) GraphMaps' visualization~\cite{mondal2017}  obtained by first computing an emanation graph of grade 1 and then moving the Steiner points using various local modifications. (right) Selection of a node.}
\label{abs}
\end{figure*} 
Motivated by the ray shooting idea that the competition mesh used, we introduce 
 a new, general \textit{$t$-spanner} called {\it the emanation graph}. An emanation graph of grade $k$, is obtained by shooting $2^{k+1}$ rays around each given point. Given a set $P$ of $n$ points in the plane, an emanation graph $M_k$ is constructed by shooting $2^{k+1}$ rays from each point $p\in P$ with equal $\frac{\pi}{2^k}$ angles between them. Each ray stops as soon as it hits another ray of shorter length or upon reaching the bounding box $R(P)$, where the lengths are computed using $L_2$  distance metric. If two parallel rays coming from opposite directions collide, then they both stop. If two rays with equal length collide at a point, then one of them is randomly stopped. The competition mesh is thus the emanation graph of grade 1. The Steiner points are created at the intersection point of the rays.  Figure~\ref{fig:simplified}~(left) and~(middle) depict  emanation graphs of grade 1 and 2 with six points in the plane, respectively.  

\begin{figure}[h]
\centering
\includegraphics[width=\textwidth,trim={0  0.7cm 0 0},clip]{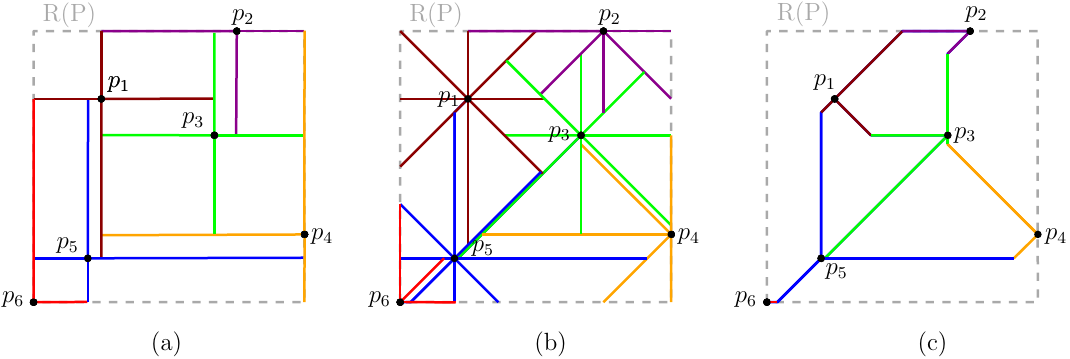}
\caption{(left) An emanation graph of grade 1. (middle) An emanation graph of grade 2. (right) A  simplified version of the emanation graph of grade 2.}
\label{fig:simplified}
\end{figure}

\subsection{Contributions}

In this paper we prove a $\sqrt{10}$ upper bound on the spanning ratio of emanation graphs of grade one, which improves the previously known upper bound of $(2+\sqrt{2})$~\cite{mondal2017}. In contrast, we prove that for $k=1$ (resp., $k>1$), there exist emanation graphs of grade $k$ with spanning ratio ratio  $\sqrt{10}$ (resp., arbitrarily close to  $\sqrt{2}$). We also show that for every fixed $k$, the emanation graphs of grade $k$ are constant spanners, where the constant factor depends on $k$.

Emanation graphs of larger grades allow many redundant edges and Steiner points, i.e., elements that can be removed without increasing the spanning ratio. Redundant edges make  a spanner  visually cluttered and unsuitable for visualization purposes unless we further refine the layout. We propose a simplification for the emanation graphs of grade 2 (e.g., Figure~\ref{fig:simplified}~(right)), which we refer to as {\it Simplified Emanation Graph (SEG)}.

 We compare SEGs with (constrained) Delaunay triangulations~\cite{shewchuk96b} on both real-word geospatial data and synthetic point sets. The synthetic point sets were created from small world graphs by the FMMM algorithm~\cite{fmmm}, which is a  well-known force directed algorithm to create network visualizations. The experimental results show SEG to achieve significantly smaller edge count, average degree, total edge length, and larger angular resolution, with a small increase of spanning ratio.

\subsection{Background}
In the following we describe further background literature related to the planar  spanners (both with and without Steiner points).

The literature on geometric spanners is rich and there are many approaches to construct geometric spanners and meshes. We refer the reader to~\cite{DBLP:journals/comgeo/BoseS13} and~\cite{DBLP:conf/imr/Owen98} for surveys on geometric spanners and mesh generation, respectively.  

Delaunay graphs are one of the most studied plane geometric spanners. Chew~\cite{DBLP:journals/jcss/Chew89} showed that the $L_1$-metric Delaunay graph is a  $\sqrt{10}$-spanner, which was later improved to 2.61 by Bonichon et al.~\cite{DBLP:journals/comgeo/BonichonGHP15}. There have been several attempts to find tight spanning ratio for Delaunay triangulations ($L_2$-metric Delaunay graphs)~\cite{DBLP:journals/comgeo/BoseDLSV11,DBLP:journals/dcg/DobkinFS90,DBLP:journals/dcg/KeilG92}. The currently best known upper and lower bound on the spanning ratio of the Delaunay triangulation are $1.998$~\cite{delauUpper} and $1.5932$~\cite{delauLower}, respectively.

Comparing emanation graphs with traditional spanners such as the Delaunay triangulation and its variants reveals interesting differences. While Delaunay meshes generally have better spanning ratios, there is no guarantee on the minimum angle between edges incident to the same node, \emph{i.e.} angular resolution of the resulting graph. Shewchuk~\cite{shewchuk96b} has thoroughly examined the {\it angular constraints} on Delaunay triangulations and introduced a Delaunay mesh generation algorithm which adds Steiner points to the original vertex set to increase the graph's angular resolution; however, the termination of this algorithm is not guaranteed for angular constraints over $34^{\circ}$. For an emanation graph, the angular resolution is determined by its grade $k$, and all emanation graphs of grade $k=2$ have $45^{\circ}$ angular resolution.

A {\it $\Theta_6$-graph}~\cite{thetaoriginalpaper1} is formed by partitioning the space around each vertex $v$ into six cones of equal angle, and then connecting the vertex $v$ to the bisector nearest neighbor in each cone; the bisector nearest neighbor in a cone means the vertex whose projection on the bisector of the cone is closest to $v$. A {\it half-$\Theta_6$-graph} is a plane geometric spanner, which is constructed in the same way except that the neighbors are only considered in the first, third and fifth cones (for some fixed clockwise ordering of the cones).  The half-$\Theta_6$-graphs were introduced by Bonichon et al.~\cite{bonichonTD}. They showed that half-$\Theta_6$-graphs have interesting connections to triangular-distance Delaunay triangulations~\cite{DBLP:journals/jcss/Chew89}, which implies that half-$\Theta_6$-graphs are 2-spanners~\cite{bonichonTD}.

While both the Delaunay triangulations and half-$\Theta_6$-graphs have a linear number of edges and small spanning ratio, they may have vertices with {\it unbounded} degree. Bose et al.~\cite{DBLP:journals/algorithmica/BoseGS05} showed that plane $t$-spanners of bounded degree exist (for some constant $t$). A significant amount of research followed this result, which examines the construction of bounded degree plane spanners with low spanning ratio. Some of the best known spanning ratios for spanners with maximum degree $4,6$ and $8$ are $20$~\cite{DBLP:journals/jocg/KanjPT17}, $6$~\cite{DBLP:conf/icalp/BonichonGHP10} and $4.414$~\cite{DBLP:journals/algorithmica/BoseHS18}, respectively. 

Although there exist point sets that do not admit a planar spanner of spanning ratio less than 1.43~\cite{DBLP:journals/ijcga/DumitrescuG16}, by allowing $O(n)$ Steiner points, one can obtain $(1+\epsilon)$-spanners, for any $\epsilon>0$. Arikati et al.~\cite{DBLP:conf/esa/ArikatiCCDSZ96} showed that one can construct a plane geometric $(1+\epsilon)$-spanner with $O(n/\epsilon^4)$ Steiner points. Bose and Smid~\cite{DBLP:journals/comgeo/BoseS13} asked whether the dependence on $\epsilon$ can be improved. Recently, Dehkordi et al.~\cite{DBLP:journals/jgaa/DehkordiFG15} proved that any set of $n$ points admits a planar angle-monotone graph of width $90^\circ$ with $O(n)$  Steiner points. Since an angle monotone graph of width $\alpha$ is a $\frac{1}{cos (\alpha/2)}$-spanner~\cite{DBLP:conf/gd/BonichonBCKLV16}, this implies the existence of a $\sqrt{2}$-spanner with $O(n)$ Steiner points, which may contain vertices of {\it unbounded} degree. See~\cite{abs-1801-06290} for more details on the construction of angle-monotone graphs with Steiner points.

Note that instead of choosing three cones in the half-$\Theta_6$-graph, one can connect a vertex to the bisector nearest neighbors in all the six cones, which gives rise to the full-$\Theta_6$-graphs. The concept has also been extended to full-$\Theta_r$ graphs~\cite{thetaUpper,thetaoriginalpaper1,thetaoriginalpaper2}, where the space around each vertex is partitioned into $r$ cones of equal angle $\theta=2\pi/r$.  Similarly, there exist Yao-graphs $Y_r$~\cite{DBLP:journals/siamcomp/Yao82,DBLP:journals/dmaa/DamianR12}, where the nearest neighbor in a cone is chosen based on the Euclidean distance. However, all these  generalizations yield   non-planar  spanners.  Researchers have also examined $\Theta$-graphs and Yao-graphs with fewer than six cones, e.g., it is known that $\Theta_4$, $\Theta_5$,   $Y_4$ and $Y_5$ graphs are  spanners~\cite{barba2013stretch,DBLP:journals/comgeo/BoseMRV15,DBLP:journals/ijcga/BoseDDOSSW12,DBLP:journals/jocg/BarbaBDFKORTVX15} but $\Theta_p$ and $Y_q$ graphs are not spanners for any $p<q<4$~\cite{el2009yao}. 

\section{Spanning Ratio of Emanation Graphs}
In this section we present some upper and lower bounds on  the spanning ratio of emanation graphs. We first prove the upper bounds in Sections~\ref{ub1}--\ref{ub2} and then prove the lower bounds on Section~\ref{lb}.

\subsection{Emanation Graphs of Grade One}
\label{ub1}

The following theorem shows a $\sqrt{10}$ upper bound  on the spanning ratio of an emanation graph.

\begin{theorem}
\label{thm:ubg1}
The spanning ratio of every emanation graph of grade one is at most  $\sqrt{10}\approx 3.162$.
\end{theorem}
\begin{proof}
Let $s$ and $t$ be a pair of vertices in the emanation graph of grade one. Consider four cones around $s$, where the cones are determined by two lines passing through $s$ with slopes $+1$ and $-1$, respectively, as illustrated in   Figure~\ref{fig:ub}; recall the line-segments that are part of the emanation graph are horizontal and vertical. Without  loss of generality assume that $t$ lies in the rightward cone $C$  of $s$. 

\begin{figure}[h]
\centering
\includegraphics[scale=1]{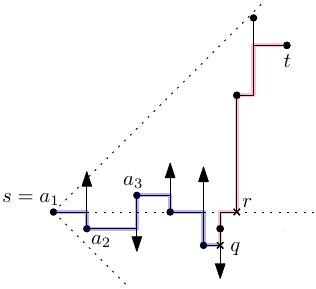}
\caption{Illustration for the proof of Theorem~\ref{thm:ubg1}.} \label{fig:ub}
\end{figure}
We now construct an $x$-monotone path $P_x$, which lies entirely in cone $C$, as follows: The path starts at $s$ and for each original vertex in this cone, the path follows its rightward segment $\ell$. If a rightward segment is stopped by another segment $\ell'$, then the path follows $\ell'$ to the original vertex  that created $\ell'$, and  continues to follow the rightward segment of this vertex. Note that $P_x$ stops at a point on the right boundary of $R(P)$, the bounding box.  Figure~\ref{fig:ub} illustrates a subpath $a_1,a_2,\ldots, q$ of $P_x$ in blue; here $s=a_1$. For any subpath $a_i,\ldots,a_j$ on $P_x$, we will use the notation $Y_{a_ia_j}$ (resp., $X_{a_ia_j}$) to refer to the sum of  the lengths of all the vertical (resp., horizontal) segments in $a_i,\ldots,a_j$. 

By construction of $P_x$ and the definition of the emanation graph, the length of any horizontal segment on $P_x$ is at least as large as the subsequent vertical segment. Hence for every subpath $a_i,\ldots,a_j$ in $P_x$, which  starts with a horizontal segment, we will have $X_{a_ia_j}\ge Y_{a_ia_j}$.

Without loss of generality assume that $t$ lies on or above $P_x$. We now construct another path $P_y$ starting at $s$ using the same construction as that of $P_x$, but following the upward segments instead of rightward ones. Note that $t$ is now in the region bounded by the paths $P_x$ and $P_y$. We now construct a directed path, called the \emph{$(-x,-y)$-monotone path} $P_t$ starting at $t$, which is in the reverse direction of the $(x,y)$-monotone path.  $P_t$ starts at $t$ and follows the leftward segment.  Since $t$ lies in the region bounded by $P_x$ and $P_y$, the path $P_t$ must intersect one of these two paths. If the last segment $\ell$ of $P_t$ is stopped by a horizontal (resp., vertical)  segment $\ell'$, then we follow $\ell'$ towards the leftward (resp., downward) direction. 
Note that $P_t$ now either intersects $P_x$ or $P_y$. Hence we consider the following two cases.

Case 1 ($P_t$ intersects $P_x$ at point $q$): This case is illustrated in Figure~\ref{fig:ub}.  
Let $\ell_h$ be the horizontal line through $s$. 

Assume first that $t$ lies above  $\ell_h$ and $q$ lies below $\ell_h$. Let $r$ be the rightmost intersection point of $P_t$ with $\ell_h$. Thus the sum of the length of the subpath of $P_x$ from $s$ to $q$ and the subpath of $P_t$ from $q$ to $t$ is as follows: 

\begin{align*}
|sq|_x+ Y_{sq} + |qt|_x + |qt|_y   
=& (|sq|_x+ |qt|_x) + Y_{sq}  + |qt|_y  \\
=& |st|_x + Y_{sq}  + |qt|_y, && \text{i.e., } |st|_x = |sq|_x+ |qt|_x \\
=& |st|_x + Y_{sq} + |qr|_y + |rt|_y, && \text{i.e., } |qt|_y = |qr|_y + |rt|_y &\\
\le&  2|st|_x + Y_{sq} + |rt|_y, && \text{i.e., } |qr|_y \le |st|_x \\
\le& 2|st|_x + |st|_x + |rt|_y, && \text{i.e., } Y_{sq}  \le |st|_x \\ 
=& 3|st|_x + |st|_y, && \text{i.e., } |rt|_y  \le |st|_y. \\
\end{align*}
Here $|st|_x$ (resp. $|st|_y$) denotes the horizontal (resp. vertical) distance between $s$ and $t$. Therefore, the spanning ratio is:
\[f=\frac{(3|st|_x + |st|_y)}{\sqrt{(|st|_x)^2 + (|st|_y)^2}}.\]
To find an upper bound we need to maximize $f$. By setting $|st|_x = 3|st|_y$, the maximum for $f$ obtains, \emph{i.e.}, $f\leq \sqrt{10} \approx 3.162$.

Assume now that $t$ and $q$ both lie on the same side of $\ell_h$.  Hence  the sum of the lengths of the paths  from $s$ to  $t$ is $
|st|_x + Y_{sq}  + |qt|_y \le  2|st|_x   + |rt|_y$. If $t$ and $q$ are below $\ell_h$, then $|qt|_y \le |st|_x$. If they are above  $\ell_h$, then $|qt|_y \le |st|_y$. Hence the path length is bounded by $3|st|_x + |st|_y$ and the upper bound of $\sqrt{10}$ holds.

Case 2 ($P_t$ intersects $P_y$ at point $q$): 
This case would be the same as when $P_t$ intersects $P_x$ with $t$ lying on the upward cone of $s$. However, applying the same analysis, we again get an upper bound of $(3|st|_x + |st|_y)$ on the length of the path $s,\ldots,q,\ldots,t$, and hence an upper bound of $\sqrt{10}$.  \qed
\end{proof}
\begin{figure}[h]
\centering
\includegraphics[scale=1]{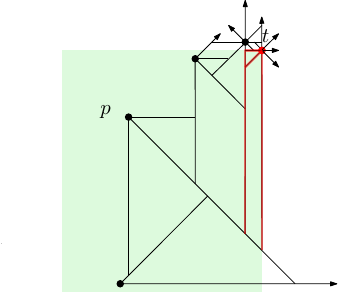}
\caption{The  maximal $(-x,-y)$-monotone paths from $t$ are shown in red.  } \label{fig:ub1}
\end{figure}
Note that the above  upper bound proof does not hold for   emanation graphs of grade 2, as the required monotone paths may not exist. For example, Figure~\ref{fig:ub1} depicts a scenario where we cannot extend an $(-x,-y)$-monotone path from $t$ to reach the bottom boundary of $R(P)$.

\subsection{Emanation Graphs of Grade $k$}
\label{ub2} 

In this section, we prove an upper bound on the spanning ratio of emanation graphs of grade $k$.   The proof will rely on the concept of angle-monotone paths. A polygonal path is an \emph{angle-monotone path of width $\gamma$} if  the vector of every edge  lies in a closed wedge of angle $\gamma$ 
(Figure~\ref{fig:am}~(left)). In other words, there exists an angle $\beta$ such that every edge vector is between $\beta+\frac{\gamma}{2}$ and $\beta-\frac{\gamma}{2}$.  Every angle-monotone path of width $\gamma$ is an $(\frac{1}{\cos (\gamma/2)})$-approximation of the Euclidean distance between its endpoints~\cite{DBLP:conf/gd/BonichonBCKLV16}. A geometric graph in the plane is  \emph{angle-monotone of width $\gamma$} if every pair of vertices is connected by an angle-monotone path of width $\gamma$. Hence these graphs are also $(\frac{1}{\cos (\gamma/2)})$-spanners.

\begin{figure}[h]
\centering
\includegraphics[width=.85\textwidth, trim={0 0.48cm 0 0},clip]{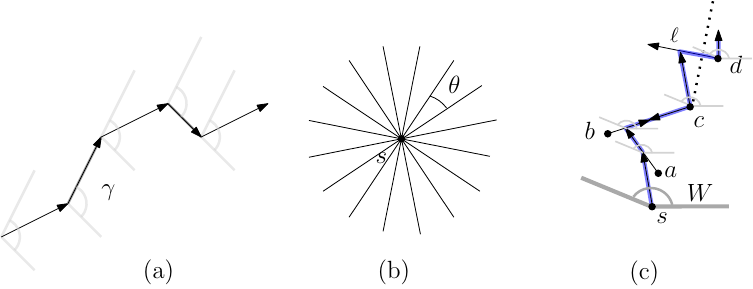}
\caption{(left) An angle-monotone path of width $\gamma$. (middle) Illustration for the cones. (right) Illustration for  $P(W)$. The dotted line illustrates the bisector of the wedge with apex at $c$.}
\label{fig:am}
\end{figure}

Let $M$ be an emanation graph with $r=2^{k+1}$ rays,  and let $s$ and $t$   be a pair of vertices in $M$. Recall that the rays around a vertex  create $r$ cones of equal angle $\theta = \frac{2\pi}{r}$. We rotate the plane by an angle of  $\theta/2$ such that no rays are axis aligned, e.g., see Figure~\ref{fig:am}~(middle). 
Let $W$ be an upward wedge of angle $(\pi-\theta)$ with apex at $s$ such that one side is aligned along the  horizontal line passing through $s$ (Figure~\ref{fig:am}~(right)). By $P(W)$ we denote a path inside $W$ that starts from the apex of $W$  and continues as follows: If a segment $\ell$ stops the last segment of the current path, then we move towards the direction which is monotone with respect to the bisector of $W$ (Figure~\ref{fig:am2}~(left)). If $\ell$ is perpendicular to the bisector, then we move towards the source of $\ell$ (Figure~\ref{fig:am2}~(right)). If  we reach an original vertex, then we repeat the process until we reach the bounding box $R$ of the point set.  The following lemma establishes the property that $P(W)$  
 lies inside $W$.
\begin{lemma}
The path $P(W)$ lies  inside $W$.
\end{lemma}
\begin{proof}
Let $p_1(=s),p_2,\ldots,p_k$ be the path $P(W)$. The first segment $p_1p_2$  of $P(W)$ is clearly inside $W$. Assume now that the segments  $p_{i-1}p_{i}$, where $2\le i\le k-1$, lie inside $W$. We now consider the segment $p_{k-1}p_{k}$. 

If $p_{k-1}$ is an original vertex, then by construction, $p_{k-1}p_k$ is a segment inside the wedge of ${p_{k-1}}$, and hence it is inside $W$. We now consider the case when  $p_{k-2}p_{k-1}$ is stopped by a segment $\ell$.

If the vector of $\ell$ is inside the 
wedge of ${p_{k-1}}$, then we route $P(W)$ along  $p_{k-1}p_{k}$ such that it is monotone with respect to the bisector (Figure~\ref{fig:am2}~(left)). Consequently, $p_{k-1}p_{k}$ lies inside $W$. 

If the vector of $\ell$ is outside of the wedge of ${p_{k-1}}$, then  $\ell$ must be perpendicular to the bisector of the wedge of ${p_{k-1}}$  (Figure~\ref{fig:am2}~(right)). Here we route $P(W)$ towards the source $r$ of $\ell$. The smallest angle that $p_{k-2}p_{k-1}$ can make with the sides of the wedge of ${p_{k-2}}$ is $\theta/2$. Since $rp_{k-1}$ is shorter than $p_{k-2}p_{k-1}$, the segment $p_{k-1}p_{k}$ must lie inside the wedge of ${p_{k-2}}$ and hence also inside $W$.      \qed 
\end{proof} 

\begin{figure}[pt]
\centering
\includegraphics[width=.75\textwidth]{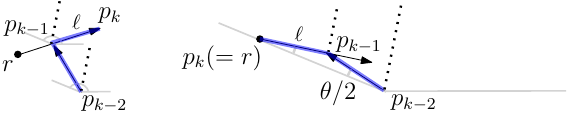}
\caption{(left) Illustration for $\ell$ when the corresponding vector lies inside the wedge of ${p_{k-1}}$.  (right) A scenario when $\ell$ is perpendicular to the bisector.}
\label{fig:am2}
\end{figure}

We are now ready to describe the construction of a path between a pair of vertices $s$ and $t$. We first define wedges $W_1,W_2,\ldots, W_r$ around $s$ (Figure~\ref{fig:case2-ub}), where $W_1$ coincides with $W$ and the subsequent wedges are obtained by rotating $W$ counter clockwise by an angle of $\theta$. Let $P(W_1), P(W_2) \ldots P(W_r)$ be the corresponding angle monotone paths of width $(\pi-\theta)$. Without loss of generality assume that $t$ is at the rightward cone $C$ of $s$, i.e., $C$ contains the positive x-axis (Figure~\ref{fig:case2-ub2new}). Let $P(W_j)$ and $P(W_{j+1})$ be a pair of angle monotone paths, where $1\le j\le r$ and $W_{r+1} = W_1$. Note that both of these paths end at the bounding box $R$ of the point set. Let $S_{j,j+1}$ be the region bounded by $P(W_j)$, $P(W_{j+1})$ and $R$. Note that the paths $P(W_j)$ and  $P(W_{j+1})$  may intersect multiple times. We now consider two cases depending on whether there exists some $i$, where $1\le i\le r$, such that $t$ lies in $S_{i,i+1}$.



\begin{figure}[h]
\centering
\includegraphics[width=.7\textwidth]{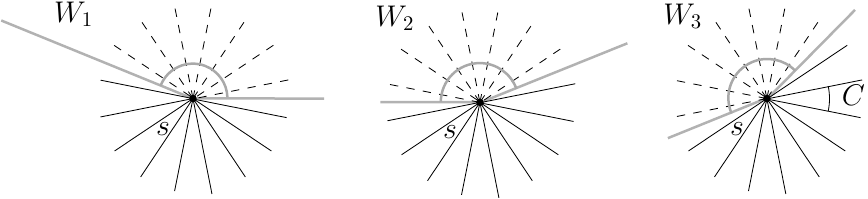}
\caption{Illustration for the wedges.}
\label{fig:case2-ub}
\end{figure}
\textbf{Case 1 (There exists a region $S_{i,i+1}$ that  contains $t$):} 
 Since the wedges $W_i$ and $W_{i+1}$ are consecutive, the right side of $W_i$ and the left side of $W_{i+1}$ lie on the same line. Let $W_t$ be the  wedge of angle $(\pi-\theta)$ at $t$ that intersects the right side of $W_i$ at a point $b$ and the left side of $W_{i+1}$ at a point $a$ where $\angle tab = \angle tba = \theta/2$. Figure~\ref{fig:case2-ub2new}(left) illustrates an example where $W_i$ and $W_{i+1}$ are shown in blue and green, respectively. Figure~\ref{fig:case2-ub2new}(right) illustrates $W_t$  in gray.

Since $t$ belongs to $S_{i,i+1}$, it suffices to consider the following two scenarios to construct a path between $s$ and $t$. 

 \textit{Case 1.1:} The path $P(W_t)$ intersects either $P(W_i)$ or $P(W_{i+1})$ at some point $q$ inside the triangle $\Delta tab$. We now use the path $P'=(s,\ldots,q,\ldots,t)$ to compute an upper bound on the spanning ratio. Here the length of $s,\ldots,q$ is at most the length of an angle monotone path of width $(\pi-\theta)$ from $s$ to $q$, plus the distances travelled along the segments that are perpendicular to the bisector of the wedges. Thus the total length is bounded by twice the length of an angle monotone path of width $(\pi-\theta)$ from $s$ to $q$.  Since the length of an angle-monotone path of width $\gamma$   between two points $a,b$ is at most $\frac{d_E(a,b)}{\cos (\gamma/2)}$~\cite{DBLP:conf/gd/BonichonBCKLV16}, the length of $s,\ldots,q$ is at most $\frac{2d_E(s,q)}{\cos (\pi/2-\theta/2)}$.  
 Since $\Delta tab$ is an isosceles triangle, $d_E(s,q)\le d_E(a,b) \le \frac{2h}{\tan{(\theta/2)}}$, where $h$ is the perpendicular distance from $t$ to $ab$. Since $h\le d_E(s,t)$, we have $\frac{2d_E(s,q)}{\cos (\pi/2-\theta/2)} \le \frac{4d_E(s,t)}{\tan{(\theta/2)}\cos (\pi/2-\theta/2)}$. Hence, if $k$ (equivalently, $\theta$) is fixed, the length of $s,\ldots,q$ can be expressed as $\delta d_E(s,t)$, where $\delta$ is a constant. Using a similar analysis for $P(W_t)$ one can show the length of $q,\ldots,t$ to be bounded by  $\delta' d_E(s,t)$, where $\delta'$ is a constant. Consequently, the length of the path $P'$ is at most $(\delta+\delta')d_E(s,t)$.

\begin{figure}[h]
\centering
\includegraphics[width=\textwidth]{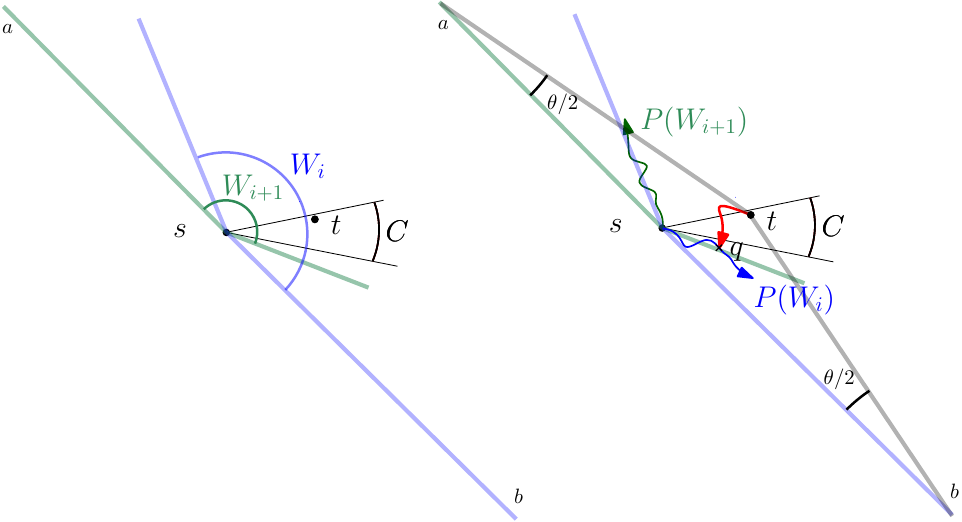}
\caption{Illustration for Case 1. (left) The wedges $W_i$ and $W_{i+1}$. (right) The construction of an $s$ to $t$ path.}
\label{fig:case2-ub2new}
\end{figure}
\textit{Case 1.2:} The path $P(W_t)$ intersects the bounding box $R$ at point $p$ inside triangle $\Delta tab$. Since $p$ is also inside $S_{i,i+1}$, there exists an orthogonal segment $pq$ inside $\Delta tab$ that intersects either $P(W_i)$ or $P(W_{i+1})$  at  point $q$. We now can use the path $P'=(s,\ldots,p,q,\ldots,t)$ to compute an upper bound on the spanning ratio.
Since $\theta\le  \pi/2$ the largest segment in $\Delta tab$ is $ab$, which is of length at most $\frac{2d_E(s,t)}{\tan{(\theta/2)}}$. Therefore, using an analysis similar to Case 1.1, we can express the length of $P'$ as $\delta'' d_E(s,t)$, where $\delta''$ is a constant.



\textbf{Case 2 (There does not exist any region $S_{i,i+1}$ that  contains $t$):} 
In this case, for every wedge $W_i$ containing  $t$, the vertex $t$ lies on the same side of $P(W_i)$.  Recall that $t$ is in the rightward cone $C$ of $s$ which contains the positive x-axis. Therefore, either $W_1$ or $W_{r/2+2}$ contains $t$. Without loss of generality assume the wedge $W' = W_1$ contains $t$. We now consider two scenarios depending on whether $t$ lies above or below the path $P(W')$.

\begin{figure}[h]
\centering
\includegraphics[width=\textwidth]{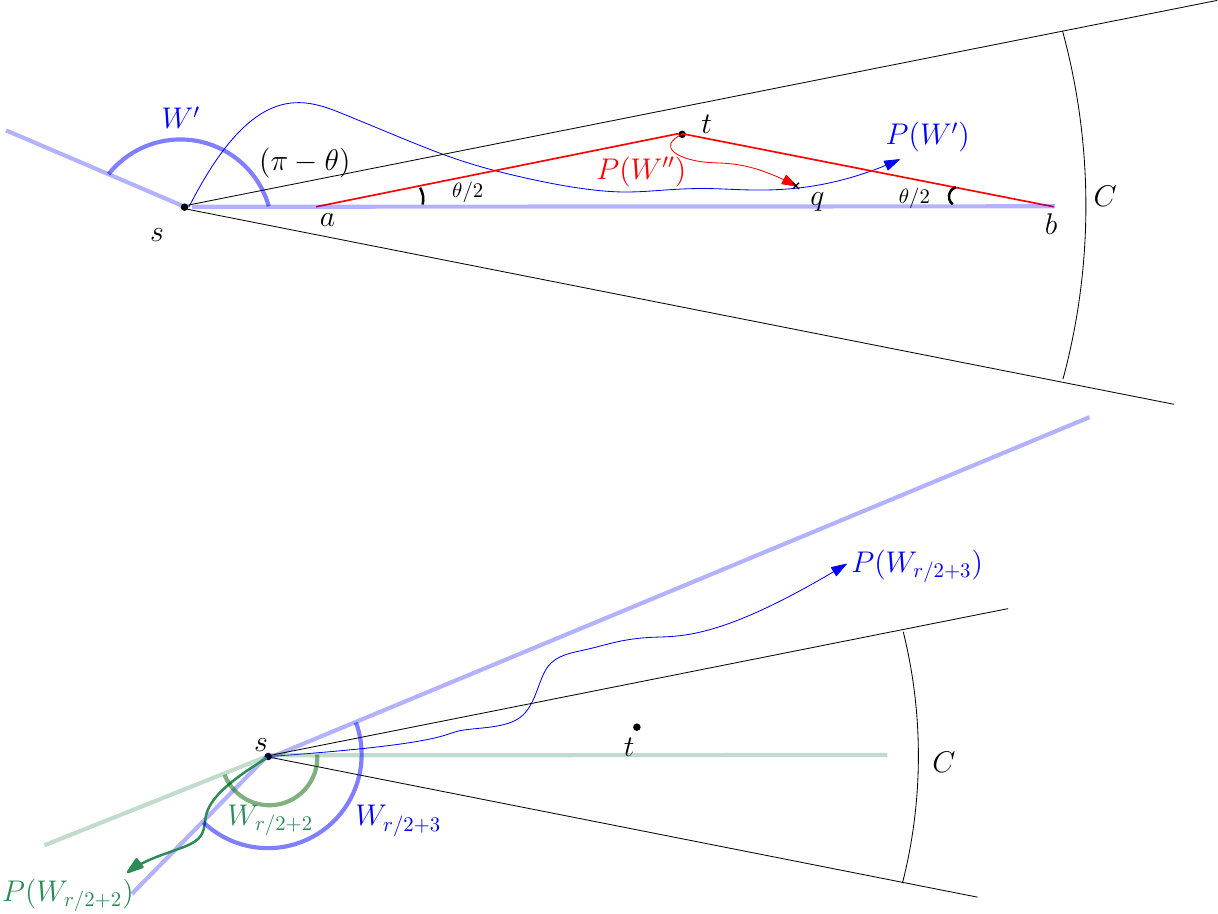}
\caption{Illustration for Case 2. (top) The construction of an $s$ to $t$ path. (bottom) Illustration for the wedges $W_{r/2+2}$ and $W_{r/2+3}$.}
\label{fig:split}
\end{figure}

\textit{Case 2.1:} If $t$ lies above the path $P(W')$, then consider a downward wedge $W''$ with apex at $t$ of angle $(\pi-\theta)$ such that the bisector of $W''$ is perpendicular to the x-axis. Let $a$ and $b$ be the point of intersections of $W'$ with the x-axis where $a$ is to the left of $b$. 
Figure~\ref{fig:split} (top) illustrates such a scenario.

First consider the case when $P(W')$ intersects $P(W'')$ at a point $q$. 
Since $\theta\le \pi/2$, we have $d_E(s,q)\le d_E(s,b) \le 2|st|_x\le 2d_E(s,t)$. Similarly, $d_E(t,q)\le d_E(t,a) \le d_E(s,t)$. We now can use the analysis of Case 1.1 to first bound the length of the paths $s,\ldots,q$ and $q,\ldots,t$, and then   show that the total length is at most a constant times $d_E(s,t)$. The case where $P(W')$ intersects the bounding box $R$ inside $\Delta tab$ can be handled in the same way as in Case 1.2.


\textit{Case 2.2:} If $t$ lies below the path $P(W')$, then we can find two successive cones $W_i$ and $W_{i+1}$ such that $P(W_i)$ and $P(W_{i+1})$ enclose $t$, as follows.

Recall our assumption in Case 2 that for every wedge $W_i$ containing  $t$, the vertex $t$ lies on the same side of $P(W_i)$. Since $t$ lies below the path $P(W')$, it must also lie below the  path $P(W_{r/2+3})$ (Figure~\ref{fig:split} (bottom)). However, since $t$ lies in $W_1=W'$, the adjacent wedge $W_{r/2+2}$ does not contain $t$ and thus $t$ would lie above $P(W_{r/2+2})$. 
Hence we can find a region $S_{r/2+3,r/2+2}$ that contains $t$, which contradicts the assumption of Case 2.


The following theorem summarizes the result of this section. 
\begin{theorem}
\label{thm:ub2}
For every fixed $k$, an emanation graph of grade  $k$ is a constant spanner, where the constant factor depends on the value of $k$.  
\end{theorem}

Note that Theorem~\ref{thm:ub2}  is only of theoretical interest as the constant factor we obtain are very large. 

\subsection{Lower Bound}
\label{lb}
The following theorem proves a lower bound on the spanning ratio of the emanation graphs.


\begin{figure}[h]
\centering
\includegraphics[width=\textwidth, trim = {0 .4cm 0 0},clip]{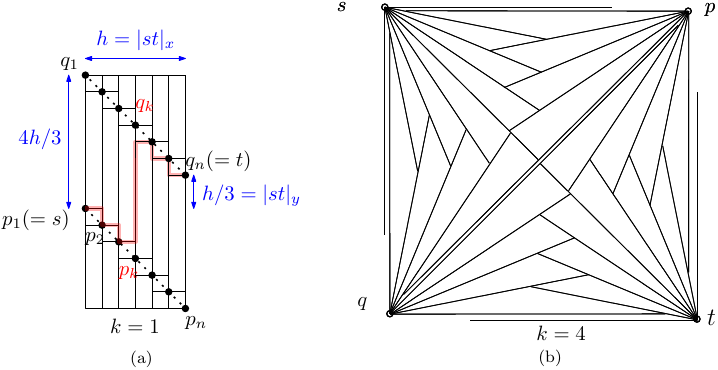}
\caption{Illustration for lower bound proof.} \label{fig:lb}
\end{figure}

\begin{theorem}
\label{the:lb}
There exists an emanation graph of grade $1$ with spanning ratio arbitrarily close to $\sqrt{10}$. 
For every $k\geq 2$, there exists an emanation graph of grade $k$ with spanning ratio  arbitrarily close to $\sqrt{2}$.
\end{theorem}
\begin{proof}
We refer the reader to Figures~\ref{fig:lb}~(left)--(right), which depict the cases $k=1$ and $k=2$, respectively.


\textbf{Case 1 ($k=1$):} We construct a set $\{p_1,\dots,p_n,q_1,\dots,q_n\}$ of $2n$ points as follows. The two points $p_1 (=s)$ and $q_n (=t)$, which will achieve the lower bound, are lying at (0,0) and $(h,h/3)$, where $h$ is a positive integer.   Imagine two parallel guiding lines with slope $-1$ through $s$ and $t$, as shown in dashed lines. One of the two guiding lines goes through $s$ and through $t$. As shown in the figure, the top-left corner of the bounding box $R$ is determined by the intersection of the vertical line through $s$ and the guiding line that starts at $t$. The bottom-right  corner of $R$  is determined by the intersection of the vertical line through $t$ and the guiding line that starts at $s$. We now place the points equidistantly on the two guiding lines such that each guiding line contains $n$ points; see Figure~\ref{fig:lb}~(left). From this, $|st|_x=h$, $|st|_y=h/3$, and $|sq_1|_y=4h/3$. 

It is straightforward to observe from the structure of the emanation graph that a shortest path must be   $x$-monotone. Let $P$ be a simple $x$-monotone path between $s$ and $t$. For every index $i$ from 1 to $n$, let $\ell_{i}$ be the line passing through $p_i$ and $q_i$. Since $s$ and $t$ are on different guiding lines, $P$ must switch from one guiding line to the other using one of these vertical lines $\ell_1,\ldots,\ell_{n}$.
 
Assume that $P$ starts at $s$, travels towards $p_k$, for some $k$ with $1\le k\le n$, and then switches the guiding line, as highlighted in red. Then the length of the path is
\begin{align*}
&~2|p_1p_k|_x -|p_{k-1}p_k|_y  + (|p_kq_k|_y-|p_{k-1}p_k|_y-|q_kq_{k+1}|_y) + 2|q_kq_n|_x - |q_kq_{k+1}|_y  \\
&=~2|p_1p_k|_x -\varepsilon  + |sq_1|_y -2\varepsilon + 2|q_kq_n|_x -\varepsilon  \\
&=~2|p_1p_k|_x  + 4h/3 + 2|p_kp_n|_x - 4\varepsilon  \\
&=~2|st|_x  + 4h/3 - 4\varepsilon \\
&=~2h  + 4h/3-4\varepsilon \\
&=~10h/3-4\varepsilon.   
\end{align*}
Here $\varepsilon = |p_1p_2|_x$ becomes arbitrarily small as $n$ approaches infinity. 

Hence for sufficiently large $n$, the spanning ratio is at least $\frac{(10h /3)-4\varepsilon}{\sqrt{h^2 + (h/3)^2}} = \frac{(10h /3)-4\varepsilon}{\sqrt{10}h/3}   \approx \sqrt{10}$.


\textbf{Case 2 ($k\geq 2$):} 
We place four points $s,p,t$ and $q$ at the corners of a square in a clockwise order; let $S$ denote the square made by points $s,p,t$ and $q$. By definition of emanation graphs, $q$ has exactly $(2^{k+1}-4)/4 = 2^{k-1}-1$ rays strictly between its upward and rightward rays. Since this is an odd number of rays, the ray in the median position will hit $p$.

We then move $p$ and $q$ towards each other along the diagonal each by a small positive constant $\varepsilon$ and then perturb by a small positive constant $\varepsilon'$ such that they do not remain along the diagonal. 
Figures~\ref{fig:lb}~(right) illustrates such a scenario.  
Assuming $\varepsilon' < \varepsilon/\sqrt{2}$, the points $p$ and $q$ lie in the square $S$, and therefore, the rays of $s$ are blocked by the rays of $p$ and $q$. Similarly, the rays of $t$ are blocked by the rays of $p$ and $q$. Consequently, the shortest path between $s$ and $t$ must visit either $p$ or $q$, which results in a spanning ratio of $\sqrt{2}$.
     \qed 
\end{proof}


\section{Simplification for Emanation Graphs of Grade Two}

An emanation graph of grade 2 has twice the number of edges than its grade 1 counterpart, i.e., for $n$ points, there are $8n$ rays and hence  at most $8n$ Steiner points. But most of these edges are redundant. For example, it is common to find two paths of shortest length between a pair of vertices, \emph{e.g.} $p_1$ and $p_3$ in Figure~\ref{fig:simplified2}. Here we propose a  simplification technique that attempts to remove such redundancies. We refer to the resulting graph as a \emph{Simplified Emanation Graph (SEG)} of grade 2. 

\begin{figure}[h]
\centering
\includegraphics[width=.7\textwidth]{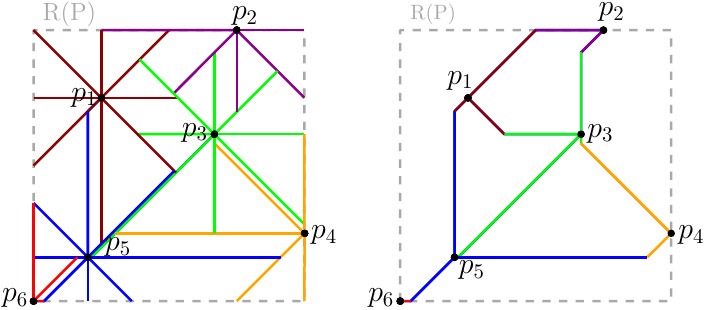}
\caption{(left) An emanation graph of grade 2 and (right) its  simplified version.}
\label{fig:simplified2}
\end{figure}
\subsection{Overview of the construction of SEG for $k=2$}
Let $G$ be an emanation graph on $n$ points with at most $8n$ Steiner points where $\beta$ of them are on the bounding box. The idea of constructing   a simplified emanation graph   is to iterate over the points and connect each point $p$ to at most 8 other points using exactly one   Steiner point per connection. The points we connect $p$ to are guaranteed to be the neighbours of $p$ in the original emanation graph. However, if we connect $p$ to a point $q$ in the SEG, then both of their rays stop at the Steiner point (whereas in original emanation graph the shorter ray would continue).  Since a Steiner point blocks a pair of rays, the number of Steiner points in SEG decreases to $(8n-\beta)/2+\beta = 4n+\beta/2$. 
 
Fix some point $p$, and assume that  $k=2$. The idea of choosing at most $8$ points for $p$ is as follows.  
Consider  $2^{k+1} = 8$  bisectors around $p$, where each bisector  bisects an angle defined by two consecutive rays originated at $p$. For each cone $C$ determined by two consecutive bisectors, we find a point $p_k$ in $C$ such that a ray $r$ of $p_k$ must touch a ray  $r'$ of $p$ irrespective of the position of the other points in $C$. We call $p_k$ the \emph{key vertex} in cone $C$ but do not create the connection between $p$ and $p_k$ immediately. The reason is that a point outside of $C$ may interfere and block $r$ or $r'$. We first select a set of \emph{candidate vertices} based on some simple distance measure from $p$, who have the potential to block the connection between $p$ and $p_k$. We then check whether any of these candidate vertices can block the connection between $p$ and $p_k$. If not, then we make the connection between $p$ and $p_k$ by creating a single Steiner point. If we can connect $p$ and $p_k$ in this process then we refer to $p_k$ as a \emph{correct neighbour} of $p$.
 

\subsection{Detailed Construction}
Fix a point $p$. While describing the search for a correct neighbor of $p$ with respect to a fixed cone $C$, 
 we rotate the plane 
 such that the cone $C$ appears to be vertically upward. For the ease of explanation, the rightward ray of a vertex is labeled $r_1$ and its other rays are numbered counter-clockwise (see Figure~\ref{fig:psselection}).   We denote the emanated rays by $r_1,r_2,r_3,...$ and their angular bisectors by  $b_1,b_2,b_3,b_4$, respectively. 
We use the notation $C_{b_ib_j}$ to refer to the cone shaped region bounded by $b_i$ and $b_j$, and  denote by $l_{g}$ a sweep line orthogonal to the bisector $g$, starting from $p$.

During the computation of the neighbours of $p$, we will  refer to two important vertex types: key vertex ($p_k$) and candidate vertex ($p_c$); we will add an edge between $p$ and $p_k$ if there is no interference by candidate vertices $p_c$.

\textbf{Key vertex of $p$:} We define the \emph{key vertex} $p_k$ to be the  first vertex found sweeping up $p$'s top cones $C_{b_2r_3}$ and $C_{r_3b_3}$.  Figure~\ref{fig:psselection}~(left) illustrates a scenario, where two sweep lines $l_{b_2}$ and $l_{b_3}$, orthogonal to $b_2$ and $b_3$, respectively, are used simultaneously to sweep $C_{b_2r_3}$ and $C_{r_3b_3}$.  Note that  a single horizontal sweep line may not hit the correct neighbor $p_k$ to be connected to $p$, \emph{e.g.} the first point $q$ hit by the horizontal sweep line may be a vertex near $p_k$ in the same cone and one of the downward rays of $p_k$ may block the connection between $q$ and $p$ (contradicting that $q$ is the correct neighbor).  Figure~\ref{fig:psselection}~(right) illustrates an example for such cases.

\begin{figure}[h]
\centering
\includegraphics[width=.35\textwidth]{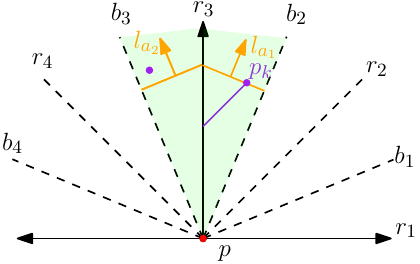}
\includegraphics[width=.35\textwidth]{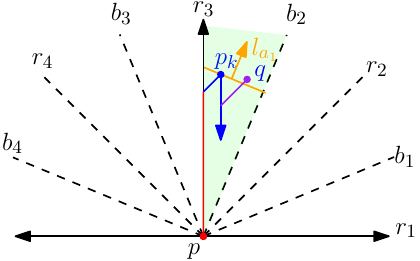}
\caption{(left) Illustration for the selection of $p_k$. Both sweep lines start at the same time from $p$ and stop as soon as one finds a vertex $p_k$. (right) An example, where a successful connection between $p$ and $p_k$ has been made, but a horizontal sweep cannot find $p_k$.}
\label{fig:psselection}
\end{figure}

\textbf{Candidate vertex of $p$:} We now consider vertices that may potentially block the connection between $p$ and $p_k$. We will impose some constraints to speed up the search. We use sweep lines with angles specific to each cone $C$ to find such {\it candidate vertex $p_c$ of $C$}. Figure~\ref{fig:pcselection} illustrates the sweep lines for each cone. The angle of the sweep line is chosen in a way so that the first vertex hit by the sweep line  wins  the competition, of reaching $p$'s connection to $p_k$, among all the points in $C$. Thus the first vertex hit by the sweep line of a cone is called the \emph{candidate vertex} $p_c$ of that cone.  The candidate vertices found by the sweep lines may block a ray of $p_k$ (see Figure~\ref{fig:pcselection}~(left)) or the upward ray of $p$ (see Figure~\ref{fig:pcselection}~(right)).  
\begin{figure}[h]
\centering
\includegraphics[width=.8\textwidth]{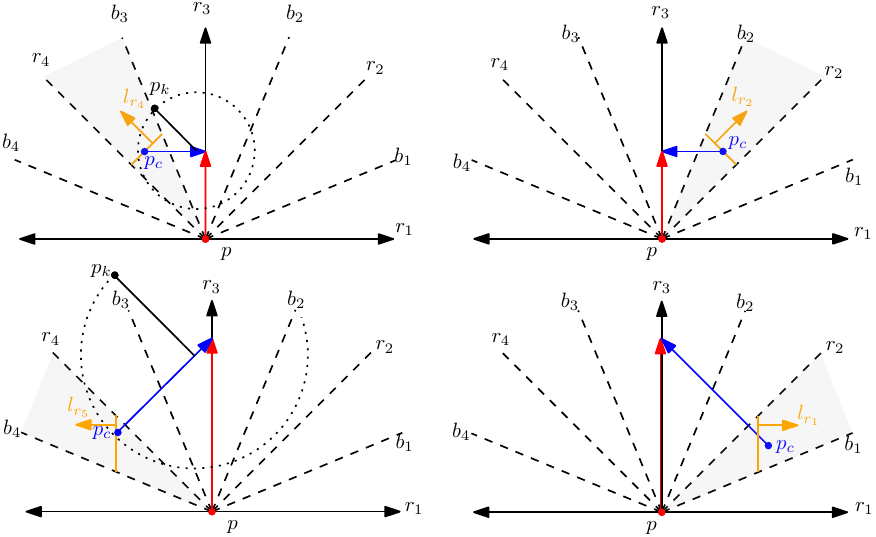}
\caption{Sweep lines used to select $p_c$ in each cone around $p$, drawn in yellow color. A sweep line starts from $p$ and stops upon finding a vertex. The dotted circles centered at the intersection point of the rays of $p_k$ and $p_c$  illustrate that $p_c$ is closer to the point of intersection than $p_k$. }
\label{fig:pcselection}
\end{figure}

We now show that to block the downward ray of $p_k$ towards $p$ or to block the upward ray of $p$, a candidate point must lie in the wedge determined by the bisectors $b_1$ and $b_4$. Thus there can only be four candidate vertices, one in each of the  four cones $C_{b_1r_2},C_{r_2b_2},C_{b_3r_4},$ and $C_{r_4b_4}$.

Without loss of generality let $q$ be a point on $r_4$. Without any interference, the rightward  ray of $q$ and the upward ray of $p$ would have the same length when they meet (e.g., see Figure~\ref{newfig}~(left)). Similarly, let $q'$ be a point on $b_4$. The ray with slope $+1$ (north-eastern) at $q'$ and the upward ray of $p$ will have the same length, i.e.,   if $o$ is the point of intersection, then $\angle opq'  (= 67.5^\circ) = \angle oq'p (= 180^\circ -45^\circ - 67.5^\circ$). Hence to block the upward ray of $p$, a point must lie in the wedge determined by the bisectors $b_1$ and $b_4$. 

\begin{figure}[h]
\centering
\includegraphics[width=\textwidth]{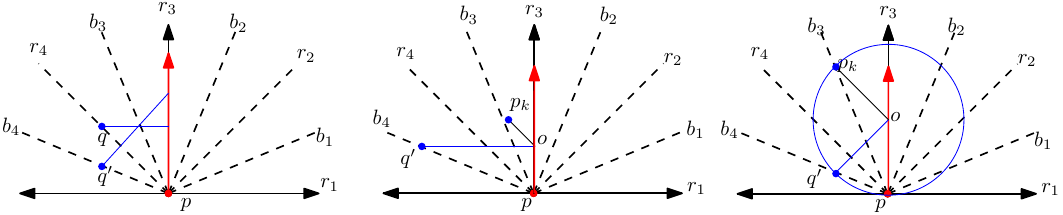}
\caption{Illustration for the  location to be searched for  candidate points.}
\label{newfig}
\end{figure}

Consider now the rightward ray of $q'$ and the downward ray (south-eastern) of $p_k$ with slope -1 (e.g., see Figure~\ref{newfig}~(middle)). Let $o$ be the point of intersection. Then $\angle oq'p_k  = 135^\circ - \angle op_kq'$. Since  $\angle op_kq'  > 90^\circ$,  the rightward ray of $q'$ must be larger than the ray of $p_k$. Finally, consider the upward ray (north-eastern) of $q'$ with slope +1 and the downward ray of $p_k$ with slope -1 (e.g., see Figure~\ref{newfig}~(right)). Let $o$ be the point of intersection. If the ray $q'o$ blocks the ray of $p_k$, then $p_ko$ must be of at least the same length as $q'o$. The length of $p_ko$ is maximized when $o$ is on the upward ray of $r$, where   the length of $p_ko$ becomes equal to the length of $q'o$. Hence to block the downward ray of $p_k$, a point must lie in the upward  wedge determined by the bisectors $b_1$ and $b_4$.



Depending on the geometric  properties of every vertex $p_c$ in a cone of $p$, some ray of $p_c$ is the most \emph{competent} (Figure~\ref{fig:sw}), meaning that it has the chance to block the connection between $p$ and $p_k$. For example, for  vertex $p_c\in C_{b_4r_4}$, the north-eastern ray $r_2$ may interfere with $p_k$, thus to find the most competent vertex inside $C_{b_4r_4}$ we use a vertical sweep line $l_{r_5}$ starting from $p$. Any point $r$, to the left of the  sweep line $l_{r_5}$ through $p_c$ inside $C_{b_4r_4}$ must have a longer ray to reach the ray of $p_k$, so it cannot block the ray of $p_k$. 

\begin{figure}[h]
\centering
\includegraphics[width=.8\textwidth]{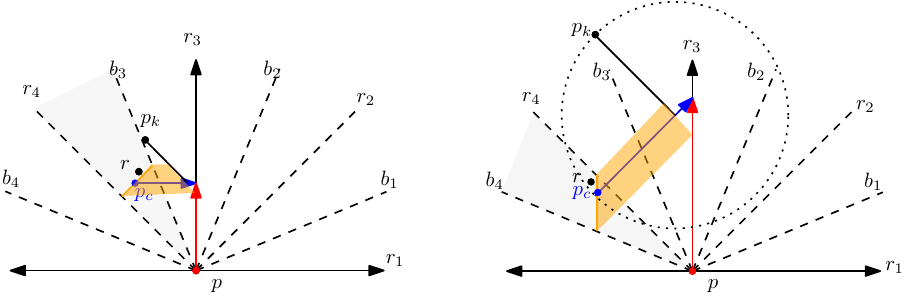}
\caption{Illustration for sweep lines. The dotted circle centered at the intersection point of the rays of $p_k$ and $p_c$  illustrates that $p_c$ is closer to the point of intersection than $p_k$.}
\label{fig:sw}
\end{figure}

After finding our candidate $p_c$ vertices, we check for some more special conditions in order to know whether they can block the connection between $p_k$ and $p$. 
 These conditions are thoroughly explained later. After the check, if a ray of $p$ can be connected to a ray of $p_k$ through one Steiner point, then we first \emph{check} whether they already have a common Steiner point neighbor. If not, we add a new Steiner point at the intersection of their rays, otherwise, we use the existing Steiner point.
%

Assuming $p_c$ lies on the right side of $r_3$, there are four cases we need to distinguish to determine whether $p_c$ interferes with the connection between $p_k$ and $p$:
\begin{itemize}
    \item Case 1: $p_k\in C_{b_2r_3}$ and $p_c\in C_{b_1r_2}$; see Figure~\ref{fig:type1}.
    \item Case 2: $p_k\in C_{b_2r_3}$ and $p_c\in C_{r_2b_2}$; see Figure~\ref{fig:type2}.
    \item Case 3: $p_k\in C_{r_3b_3}$ and $p_c\in C_{b_1r_2}$; see Figure~\ref{fig:type3}.
    \item Case 4: $p_k\in C_{r_3b_3}$ and $p_c\in C_{r_2b_2}$; see Figure~\ref{fig:type4}.
\end{itemize}

\begin{figure}[h]
\centering
\includegraphics[width=\textwidth]{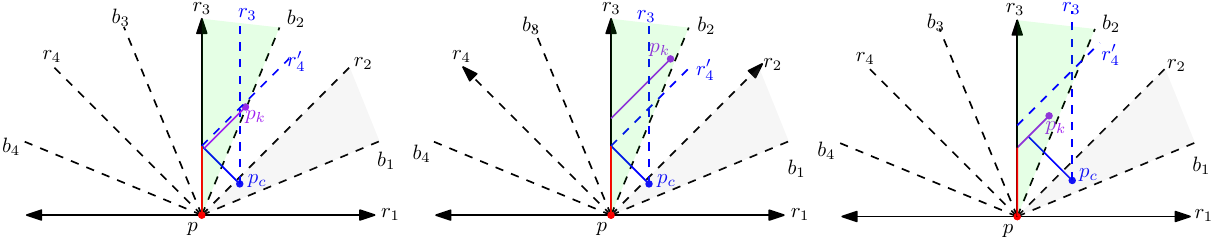}
\caption{Case 1: Left and middle depict two different cases where $p_c$ has interfered, right shows a successful connection between $p_k$ and $p$.}
\label{fig:type1}
\end{figure} 
\begin{figure}[h]
\centering
\includegraphics[width=0.8\textwidth]{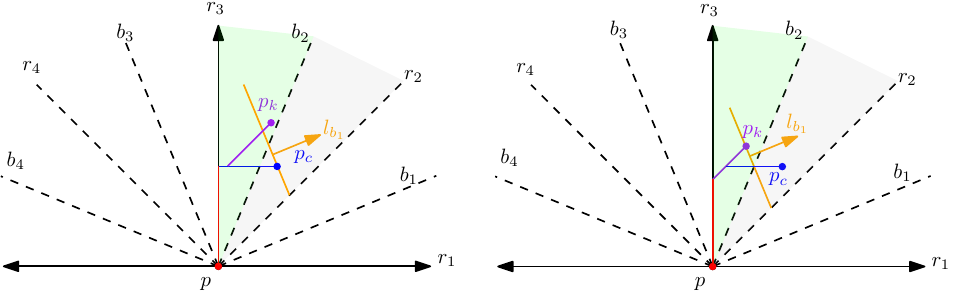}
\caption{Case 2: Left depicts $p_c$ has interfered, right shows a connection between $p_k$ and $p$.}
\label{fig:type2}
\end{figure}
\begin{figure}[h]
\centering
\includegraphics[width=0.8\textwidth]{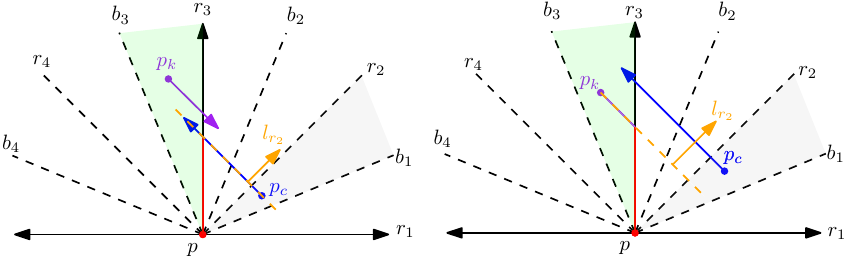}
\caption{Case 3: Left depicts $p_c$ has interfered, right shows a connection  between $p_k$ and $p$.}
\label{fig:type3}
\end{figure}
\begin{figure}[h]
\centering
\includegraphics[width=0.8\textwidth]{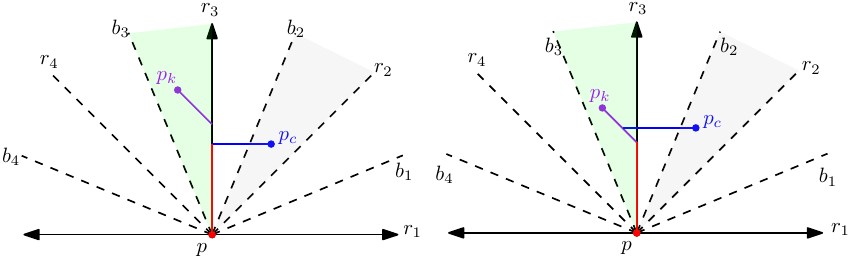}
\caption{Case 4: Left depicts $p_c$ has interfered, right shows a connection  between $p_k$ and $p$.}
\label{fig:type4}
\end{figure}

Figures~\ref{fig:type1}--\ref{fig:type4} illustrate examples for each case, where the rightmost section in each figure depicts the case when $p_k$ can successfully connect to $p$. Let $|p|_x$ (resp. $|p|_y$) be the $x$ (resp. $y$)-coordinate of the point $p$.

Let $r'_4$ be the continued refraction of $r_4$ of $p_c$ after hitting $r_3$ of $p$. In Case~1, if $|p_k|_x < |p_c|_x$ and $p_k$ is below $r'_4$, the south-western ray of $p_k$ reaches to $r_3$ sooner than the north-western ray of $p_c$. Therefore, $p_c$ cannot interfere with the connection between $p_k$ and $p$; see Figure~\ref{fig:type1}~(right). In this case, $p_c$ could block the south-western ray of $p_k$ (as shown in Figure~\ref{fig:type1}~(left)) or the upward ray of $p$ (as shown in Figure~\ref{fig:type1}~(middle)).

In Case~2, if $p_k$ is swept before $p_c$ by the sweep line $l_{b_1}$, the south-western ray of $p_k$ reaches to $r_3$ sooner than the western ray of $p_c$. This implies that $p_k$ should connect to $p$; see Figure~\ref{fig:type2}.

In Case~3, if $p_k$ is swept before $p_c$ by a sweep line $l_{r_2}$, the south-eastern ray of $p_k$ blocks $r_3$ before the north-western ray of $p_c$ reaches there, so $p_c$ cannot interfere with the connection between $p_k$ and $p$; see Figure~\ref{fig:type3}.

In Case~4, if $|p_k p_c|_y < |p_k p|_x$, the south-eastern ray of $p_k$ blocks $r_3$ before the western ray of $p_c$ reaches there; therefore, $p_k$ connects to $p$; see Figure~\ref{fig:type4}.

%

Explaining the cases where $p_c$ is on the left side of $r_3$ is straightforward, as every condition needs to be vertically mirrored, relative to $p$. 

Figure~\ref{fig:construction} demonstrates all the 8 steps (with rotations) described above using the point set used in Figure~\ref{fig:simplified2}. After stacking blue segments after rotating them back to their starting direction results into the  simplified version of the emanation graph.
\begin{figure}[h]
\centering
\includegraphics[width=1\textwidth]{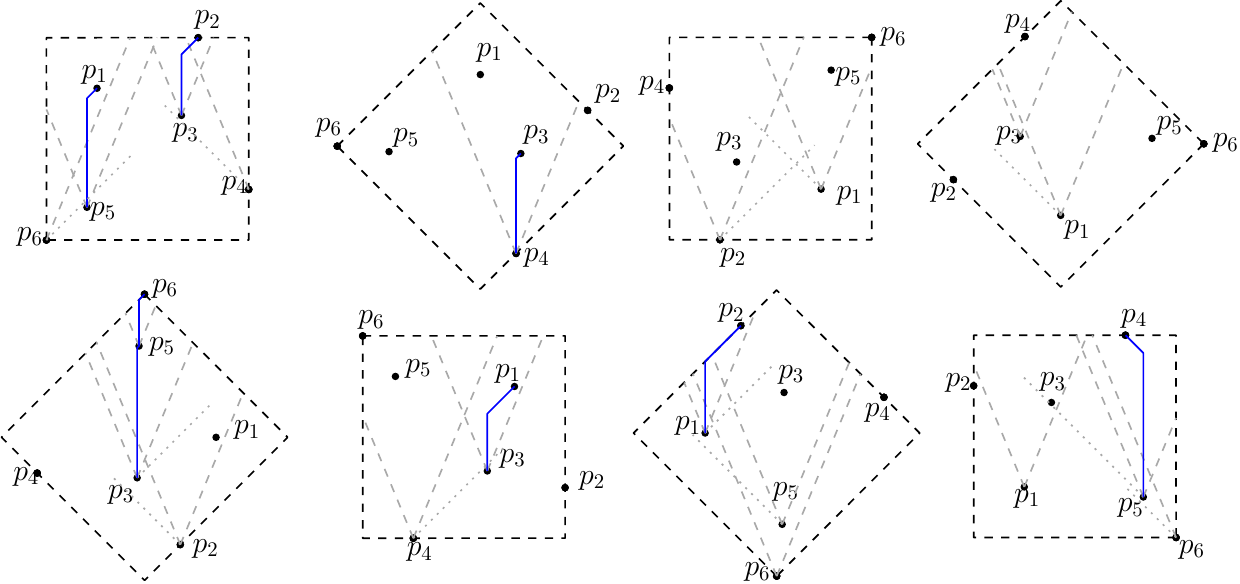}
\caption{Construction steps of an example SEG, each figure represents one of 8 required rotations. Blue segments are rotated back and accumulated to form the final graph, depicted in the right section of Figure~\ref{fig:simplified2}.}
\label{fig:construction}
\end{figure}

\section{The Construction Algorithm}
In the following section we discuss a few properties of SEG. 

\begin{lemma}
\label{th:time}
A SEG on a set of $n$ points can be constructed in time $O(n\cdot \emph{polylog}(n))$. 
\end{lemma}
\begin{proof}
For each point $p$, there exist a constant number of cones, and for each cone we need to find a candidate point with the smallest coordinate along some axis. This can easily be done by using a constant number of {\it $2$-dimensional range trees} each is corresponding to a cone, which can be constructed in time $O(n\cdot \emph{polylog}(n))$ (Theorem 5.9 in~\cite{Berg2008txtbk}). At each internal node $v'$ of the second-level trees ${\cal T}_{assoc}(v)$, we store the point with the smallest coordinate along the axis among the points in $P(v')$, where $v$ is an internal node of the first-level tree ${\cal T}$ and $P(v')$ is the set of points stored at the leaves of the sub-tree rooted at $v'$.
To find the point with the smallest coordinate along the axis of some cone, we can easily query the corresponding range tree in time  $O(\emph{polylog}(n))$.

After finding the candidates $p_c$ in the cones of each point $p$, we do a constant number of comparisons with $p_k$ in order to check whether $p_c$ has interfered the connection between $p$ and $p_k$. Therefore, the total construction time is $O(n\cdot \emph{polylog}(n))$.
\qed
\end{proof}

Forming an emanation graph of grade $k=2$ involves shooting $2^{k+1}$ rays from each vertex simultaneously. This results into a maximum degree of $8$ and $8n$ rays in any graph and $8n$ maximum number of Steiner points. Any pair of selected vertices $(p,p_k)$ in an emanation graph, falls in one of four categories:
\begin{enumerate}
\item They are not connected to each other through a single Steiner point, 
 because other vertices have completely interfered their connection; see $(p_1,p_4)$ in Figure~\ref{fig:simplified2}.
\item They are connected by two mirrored paths of two edges; 
 see $(p_1,p_5)$ in Figure~\ref{fig:simplified2}.
\item They are connected by a path of two edges, 
 and another path of longer length. The second path is formed due to interference of a ray from $p_c$ (i.e., $p_1$), thus involving an edge belonging to $p_c$; see $(p_3,p_5)$ in Figure~\ref{fig:simplified2}.
\item They are connected by a path of  two edges, but neither Category 2 nor Category 3 are satisfied; see $(p,t)$  in Figure~\ref{fig:ub}~(right). 
\end{enumerate}
A simplified emanation graph will reduce paths of categories 2 and 3, and thus will reduce Steiner points. Between path pairs of category 2, one is picked arbitrarily and another is omitted. Also for paths of category 3, the one with shorter length remains as the one with longer length is removed. Therefore, it is straightforward to construct examples where the number of Steiner points in a SEG are significantly smaller than the emanation graph (e.g., points on a line with angle of inclination $50^\circ$).

\begin{lemma}
An emanation graph of grade $k$ contains $kn$ Steiner points and there exist point sets where an emanation graph must generate   $kn-O(k)$ Steiner points. Let $G$ be an emanation graph of grade 2 on a set $P$ of $n$ points. Then $G$ contains at most $8n$ Steiner points. Assume that $\beta$ of the Steiner points are on the bounding box. Then a SEG on $P$ will contain at most $4n+\beta/2$ Steiner points. 
\end{lemma}
\begin{proof}
Since an emanation graph of grade $k$ contains $kn$ rays and each generates at most one Steiner point, the total number of Steiner points is at most $kn$. To observe that there exist point sets that generate $kn-O(k)$ Steiner points,  first place 4 points along the four corners of a square $R$ and 4 at the midpoint of its edges. We then   place $n-8$ points at its center. We perturb the points at the center to avoid overlap. We thus get $k(n-8)$ rays creating $kn-O(k)$ Steiner points inside $R$.
 
The upper bound on the Steiner points of SEG follows from the observation that every Steiner point that does not lie on the bounding box is the result of two rays hitting each other when both of them stop.\qed
\end{proof}





\section{Experimental Comparison}
\label{sec:exp}

In this section we compare SEG with graphs generated with Delaunay triangulation:  constrained~\cite{shewchuk96b} and normal. A normal Delaunay triangulation on a given set of points in general position is defined using the empty circle condition, i.e., three points form a  triangle if and only if the interior of the    circumcircle  does not enclose any point of the pointset.   The constrained Delaunay triangulation~\cite{shewchuk96b} is generated by setting a minimum angle constraint, where Steiner points are added to guarantee all angles to be above the specified constraint.

\begin{figure}[h]
\centering
\includegraphics[width=.9\textwidth]{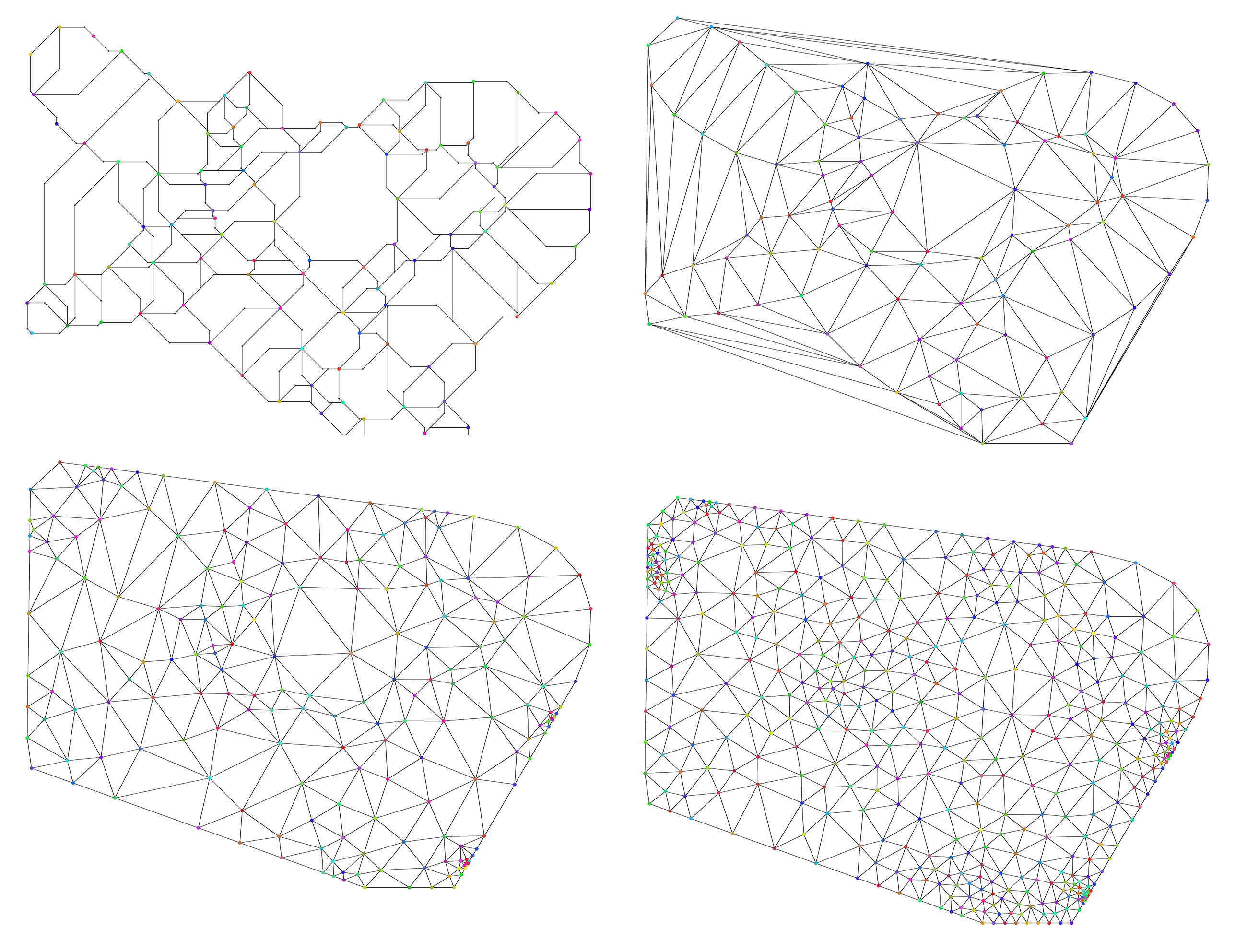}
\caption{A SEG based on our chosen sample of size 100 (top-left). The normal Delaunay triangulation (top-right), 22.5$^{\circ}$ constrained Delaunay triangulation (bottom-left) and 33$^{\circ}$ constrained Delaunay triangulation (bottom-right), all on the same vertex set.}
\label{fig:comparison}
\end{figure}
\begin{figure}[h]
\centering
\includegraphics[width=0.48\textwidth]{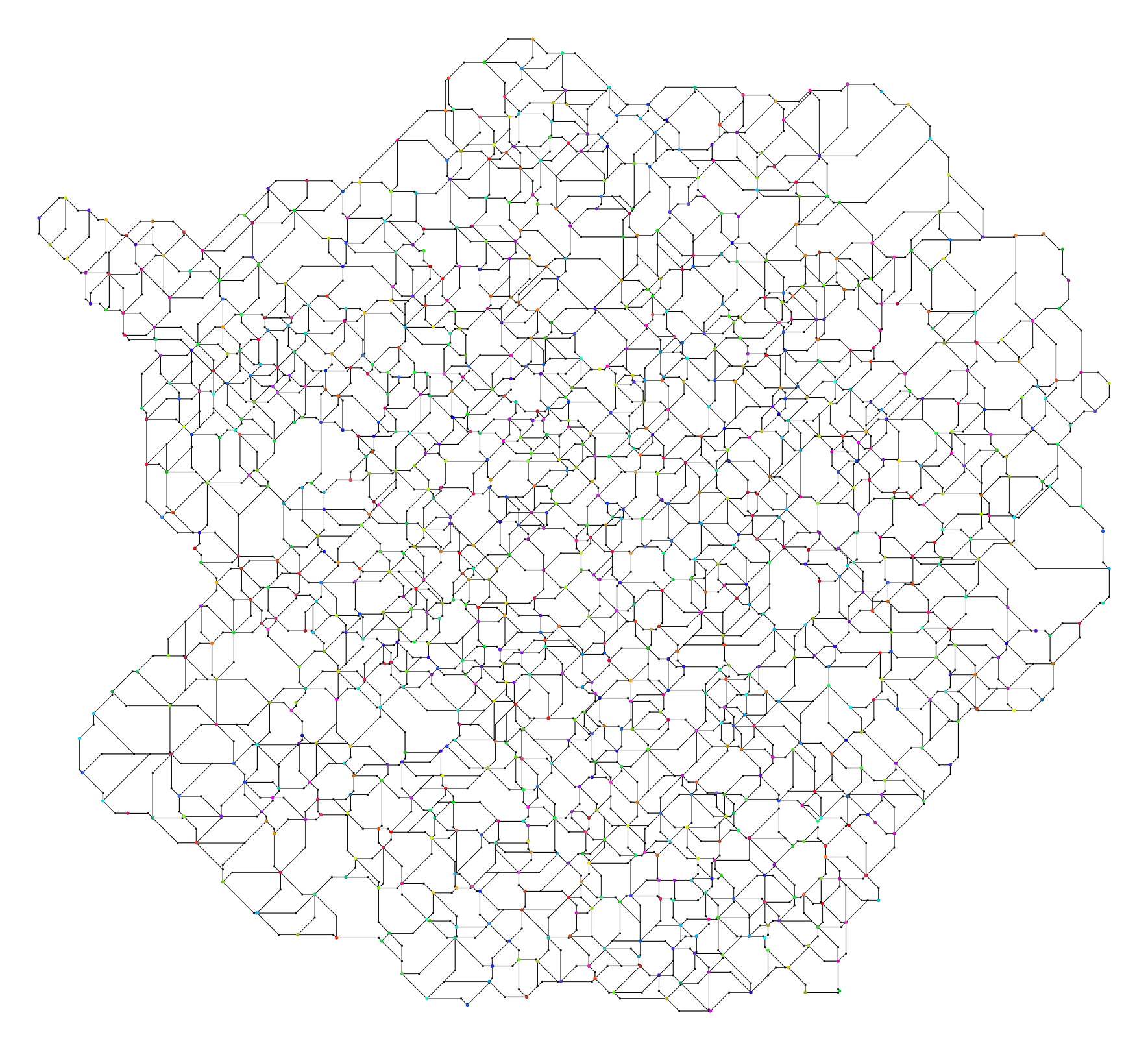}
\includegraphics[width=0.48\textwidth]{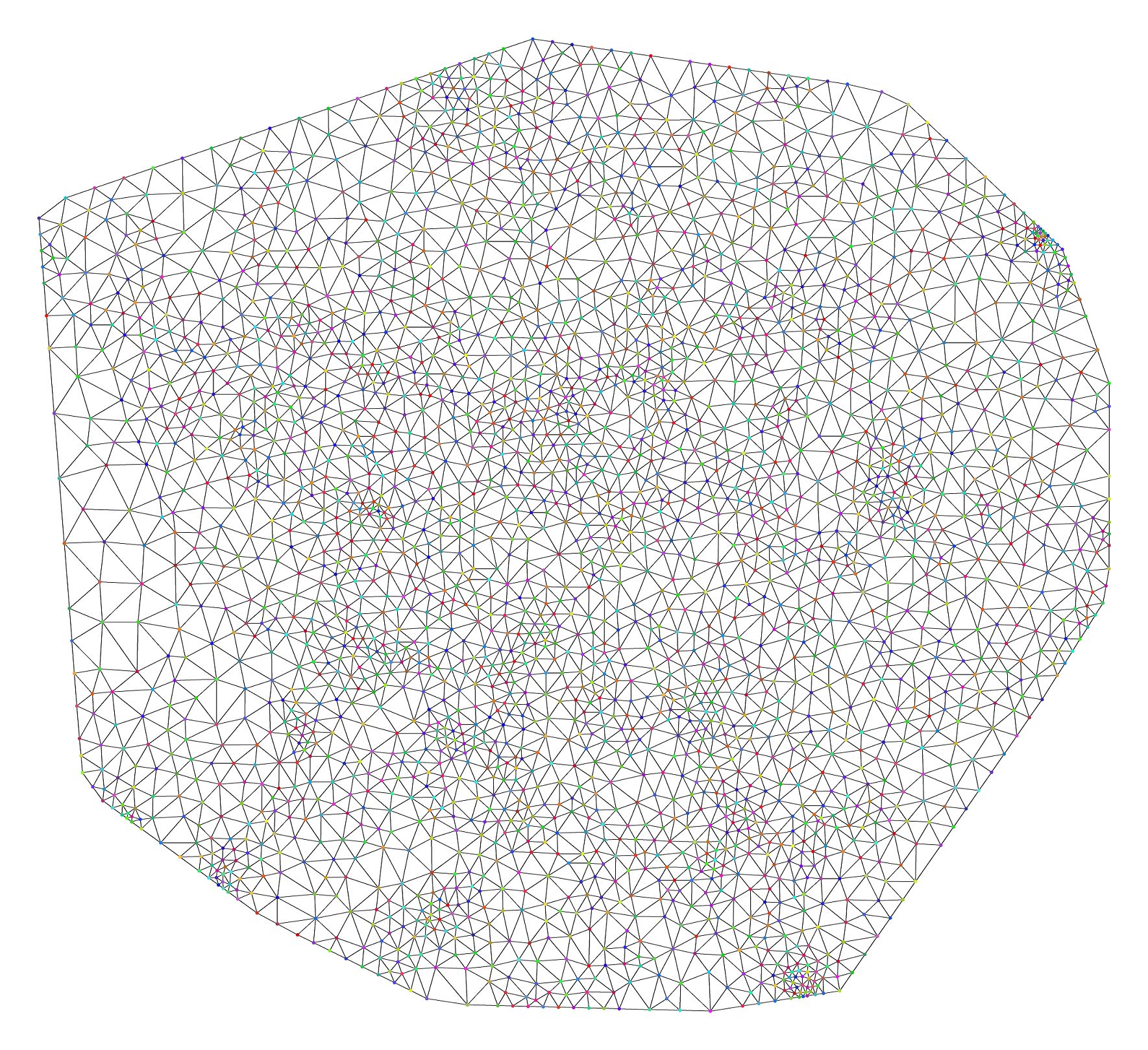}
\caption{(left) A SEG of grade 2 on a sample of size  1000. (right) The corresponding 33$^{\circ}$ constrained Delaunay triangulation.}
\label{fig:comparison3}
\end{figure}

We generated three  datasets (Rand1, Rand2 and Rand3) using NetworkX~\cite{netX}, each containing $1000$ random Newman\_ Watts\_Strogatz small world graphs. All the graphs in a data set contain the same number of nodes. Thus the three data sets contain graphs of size $100$, $500$, and $1000$. We generated the layout for all these graphs using the fast multi-pole multilevel (FMMM) layout~\cite{fmmm}. Aside from experimenting on randomly generated data~\cite{Git}, we also tried SEG on two commonly used data sets:   \emph{Locations of 1000 Most Populated Cities} and \emph{US Airports}~\cite{Git2}.

\begin{sidewaystable}[pt]
\vspace{12cm}
\noindent\makebox[.8\textwidth]{%
\addtolength{\tabcolsep}{0.5pt}
\begin{tabular}{lcc|ccccccccc}
Configuration& \multicolumn{1}{c}{\begin{sideways}Point Count\end{sideways}} & \multicolumn{1}{c}{\begin{sideways}Data Set\end{sideways}} & \multicolumn{1}{c}{\begin{sideways}Steiner Points\end{sideways}}& \multicolumn{1}{c}{\begin{sideways}Max Degree\end{sideways}} & \multicolumn{1}{c}{\begin{sideways}Average Degree\end{sideways}} & \multicolumn{1}{c}{\begin{sideways}Edge Count\end{sideways}} & \multicolumn{1}{c}{\begin{sideways}Max Edge Len\end{sideways}} & \multicolumn{1}{c}{\begin{sideways}Average Edge Length\end{sideways}} & \multicolumn{1}{c}{\begin{sideways}Total Edge Length\end{sideways}} & \multicolumn{1}{c}{\begin{sideways}Min Angle\end{sideways}} & \multicolumn{1}{c}{\begin{sideways}Spanning Ratio\end{sideways}} \\ \hline
\textit{\textbf{SEG}} & \textit{\textbf{100}} & \textit{\bf Rand1} & \textit{\textbf{197.25}} & \textit{\textbf{6.20}} & \textit{\textbf{2.55}} & \textit{\textbf{404.54}} & \textit{\textbf{80.33}} & \textit{\textbf{19.36}} & \textit{\textbf{7319.71}} & \textit{\textbf{45.00}} & \textit{\textbf{1.88}} \\
DEL C=0 & 100 & \textit{Rand1} & 0 & 9.55 & 5.66 & 283.29 & 304.15 & 50.50 & 14301.99 & 0.57 & 1.37 \\
DEL C=22.5 & 100 & \textit{Rand1} & 87.73 & 9.05 & 5.40 & 506.70 & 89.49 & 29.34 & 14803.31 & 22.68 & 1.44 \\
DEL C=33 & 100 & \textit{Rand1} & 315.60 & 8.18 & 5.56 & 1156.10 & 56.95 & 18.75 & 21392.89 & 33.07 & 1.59 \\ \hline
\textit{\textbf{SEG}} & \textit{\textbf{500}} & \textit{\bf Rand2}& \textit{\textbf{1085.44}} & \textit{\textbf{6.85}} & \textit{\textbf{2.63}} & \textit{\textbf{2262.01}} & \textit{\textbf{69.35}} & \textit{\textbf{12.58}} & \textit{\textbf{28451.72}} & \textit{\textbf{45.00}} & \textit{\textbf{2.07}} \\
DEL C=0 &  500 & \textit{Rand2} & 0 & 10.31 & 5.91 & 1478.01 & 317.07 & 28.97 & 42820.80 & 0.27 & 1.39 \\
DEL C=22.5 & 500 & \textit{Rand2}  & 253.14 & 9.48 & 5.75 & 2165.57 & 76.87 & 21.06 & 45576.34 & 22.55 & 1.60 \\
DEL C=33 & 500 &  \textit{Rand2}  & 1017.13 & 8.69 & 5.79 & 4398.15 & 47.44 & 14.35 & 62963.46 & 33.02 & 1.84 \\ \hline
\textit{\textbf{SEG}} & \textit{\textbf{1000}} & \textit{\bf Rand3}  & \textit{\textbf{2177.86}} & \textit{\textbf{7.01}} & \textit{\textbf{2.66}} & \textit{\textbf{4601.47}} & \textit{\textbf{58.10}} & \textit{\textbf{9.52}} & \textit{\textbf{43864.43}} & \textit{\textbf{45.00}} & \textit{\textbf{2.16}} \\
DEL C=0 & 1000 & \textit{Rand3} & 0 & 10.72 & 5.95 & 2974.35 & 284.05 & 21.16 & 62933.44 & 0.20 & 1.40 \\
DEL C=22.5 & 1000 &\textit{Rand3} & 472.20 & 9.75 & 5.83 & 4296.57 & 64.75 & 15.91 & 68346.21 & 22.53 & 1.96 \\
DEL C=33 & 1000 & \textit{Rand3}& 1933.36 & 8.90 & 5.86 & 8601.86 & 39.35 & 10.95 & 94099.05 & 33.01 & 2.16 \\ \hline
\textit{\textbf{SEG}} & \textit{\textbf{235}} & \textit{\textbf{AIR}} & \textit{\textbf{485}} & \textit{\textbf{7}} & \textit{\textbf{2.69}} & \textit{\textbf{970}} & \textit{\textbf{61.17}} & \textit{\textbf{9.95}} & \textit{\textbf{9651.47}} & \textit{\textbf{45.00}} & \textit{\textbf{1.96}} \\
DEL C=0 & 235 & AIR & 0 & 11 & 5.89 & 692 & 291.34 & 26.64 & 18432.64 & 0.06 & 1.39 \\
DEL C=22.5 & 235 & AIR & 221 & 9 & 5.57 & 1270 & 66.82 & 15.50 & 19682.53 & 22.53 & 1.48 \\
DEL C=33 & 235 & AIR & 729 & 8 & 5.66 & 2727 & 56.96 & 10.15 & 27691.84 & 33 & 1.73 \\ \hline
\textit{\textbf{SEG}} & \textit{\textbf{1000}} &  \textit{\textbf{CIT}} & \textit{\textbf{1913}} & \textit{\textbf{6}} & \textit{\textbf{2.59}} & \textit{\textbf{4102}} & \textit{\textbf{394.56}} & \textit{\textbf{18.14}} & \textit{\textbf{74447.77}} & \textit{\textbf{45.00}} & \textit{\textbf{2.35}} \\
DEL C=0 & 1000 & CIT & 0 & 12 & 5.95 & 2975 & 2024.00 & 50.93 & 151541.50 & 0.09 & 1.41 \\
DEL C=22.5 & 1000 & CIT & 1358 & 10 & 6.28 & 7414 & 373.59 & 27.64 & 192302.55 & 22.55 & 1.49 \\
DEL C=33 & 1000 & CIT & 4676 & 9 & 6.03 & 17139 & 166.38 & 18.34 & 308560.57 & 33.00 & 1.60 \\ 
\end{tabular}}
\centering
\caption{Results of our comparisons on 3 random and two real data sets. 
CIT marks the results related to World's most populated cities data set while AIR refers to the data set of US Airlines, lines that are unmarked are related to our random experimentation, based on averages of $1000$ instances. \textit{SEG} stands for Simplified Emanation Graph of grade $k=2$ and DEL C=$\alpha$ is a $\alpha^{\circ}$ constrained Delaunay Triangulation. 
}
\label{tab:results1}
\end{sidewaystable}

Figure~\ref{fig:comparison} demonstrates our output on one of the sample data set of size 100, for a SEG of grade 2 along with normal, 22.5$^{\circ}$ and 33$^{\circ}$ constrained Delaunay triangulations, which are the exact configurations we used for this comparison.   
Figure~\ref{fig:comparison3}  depicts SEG of grade 2 and the corresponding constrained Delaunay triangulations for a sample of size 1000.

Although one would like to have angular constraints higher than 33$^{\circ}$ and close to what emanation graph gives, the algorithm for constrained Delaunay triangulation does not guarantee termination for larger angular resolutions. We used Triangle~\cite{Triangle} to compute the Delaunay triangulations.

The metrics we chose to compare our samples are \textit{Steiner Point Count, Vertex Degree, Edge Count, Edge Length, Angle and Spanning Ratio}. Results are depicted in Table~\ref{tab:results1}, separated by different configurations and the number of vertices. 
For the first three datasets (Rand1, Rand2 and Rand3), every row of the table shows the mean performance over all $1000$ instances of the graphs. The reason that we report the averages is because the average is a better representative when examining the properties for a graph family (i.e., small world graphs) than the outcomes for individual instances. 
 In comparison with $33^{\circ}$ constrained Delaunay triangulation, SEG {shows}:
\begin{itemize}
    \item[$\bullet$] Much better angular resolution ($45^\circ$ compared to $33^\circ$)
    \item[$\bullet$] Less than half the number of edges
    \item[$\bullet$] Less than half the total edge length
    \item[$\bullet$] Less than half the average vertex degree
    \item[$\bullet$] Slightly worse spanning ratio (within a factor of 1.18 when $n =100$ and $n=500$; and the comparable when $n=1000$)
    \item[$\bullet$] Comparable number of Steiner points (less than half the number of Steiner points for $n =100$; but slightly worse for $n=1000$)
\end{itemize}

The reason that SEG provides better angular resolution than that of $33^\circ$ constrained Delaunay triangulation is inherent to its construction, where the slopes of the edges are in $\{0,\pm 1, \pm \infty\}$.  The number of edges of SEG is smaller because it has only two edges adjacent to each Steiner point whereas much more is often needed in a $33^\circ$ constrained Delaunay triangulation. 
Together with the fact that SEG does not consider filling the empty spaces around Steiner points, this significantly reduces  the total edge length and average vertex degree. The spanning ratio of SEG appears to be slightly worse. 
A potential reason is that every bend on a path at a Steiner point is of at least   $45^\circ$, whereas in a $33^\circ$ constrained Delaunay triangulation a path has an opportunity to reduce its  bend angles by leveraging the high degree of the Steiner points. The number of Steiner points in SEG appears to be smaller when the number of points is small, but it becomes slightly larger as the number of points increases. A potential reason is that two points can directly be connected in a $33^\circ$ constrained Delaunay triangulation, whereas in SEG, they must be connected through a Steiner point. Hence for a dense point set, this benefit of a $33^\circ$ constrained Delaunay triangulation may outweigh SEG.

\section{Conclusion} 
\label{sec:con}
The most obvious open question following our work is to  find a tight bound on the spanning ratio for emanation graphs of grade 2. Another interesting research direction is to find a geometric spanner that is better than the emanation graphs of grade one; specifically, a max-degree-4 planar  geometric spanners with at most $4n$ Steiner points and a spanning ratio better than $\sqrt{10}$. It would be interesting to examine whether known bounded degree  spanners~\cite{DBLP:journals/dcg/BonichonKPX15} without Steiner points could be modified to construct such a spanner. It would also be interesting to examine whether emanation graphs admit local routing with small routing ratio.

 
A natural extension of our work is to implement simplified emanation graphs in visualization systems such as  GraphMaps~\cite{DBLP:conf/gd/NachmansonPLRHC15} to compare the visual results with those generated by the Delaunay and constrained Delaunay triangulations. Although simplified emanation graphs appear to be promising in our experimental analysis, we do not know whether  they admit  a bounded spanning ratio.  Therefore, it would be interesting to further explore the spanning properties of these graphs. 

\begin{acknowledgements}
The research  D. Mondal is supported in part by the Natural Sciences and Engineering Research Council of Canada (NSERC). We thank anonymous reviewers for their helpful  suggestions and feedback to improve the presentation of the paper. 
\end{acknowledgements}

%
%


\begin{thebibliography}{}
\bibitem{Git2}
Gephi sample data sets: US airlines.
\newblock \url{https://github.com/gephi/gephi/wiki/Datasets}, 2019.
\newblock Online; accessed 6 June 2019.

\bibitem{DBLP:conf/esa/ArikatiCCDSZ96}
S.~R. Arikati, D.~Z. Chen, L.~P. Chew, G.~Das, M.~H.~M. Smid, and C.~D.
  Zaroliagis.
\newblock Planar spanners and approximate shortest path queries among obstacles
  in the plane.
\newblock In {\em Proceedings of the 4th Annual European Symposium on Algorithms (ESA)}, pages 514--528, 1996.

\bibitem{Berg2008txtbk}
M.~d. Berg, O.~Cheong, M.~v. Kreveld, and M.~Overmars.
\newblock {\em Computational Geometry: Algorithms and Applications}.
\newblock Springer-Verlag TELOS, Santa Clara, CA, USA, 3rd ed. edition, 2008.

\bibitem{DBLP:conf/gd/BonichonBCKLV16}
N.~Bonichon, P.~Bose, P.~Carmi, I.~Kostitsyna, A.~Lubiw, and S.~Verdonschot.
\newblock Gabriel triangulations and angle-monotone graphs: Local routing and
  recognition.
\newblock In {\em Proceedings of the 24th International Symposium on Graph
  Drawing and Network Visualization (GD)}, pages 519--531, 2016.

\bibitem{bonichonTD}
N.~Bonichon, C.~Gavoille, N.~Hanusse, and D.~Ilcinkas.
\newblock Connections between theta-graphs, {D}elaunay triangulations, and
  orthogonal surfaces.
\newblock In D.~M. Thilikos, editor, {\em Proceedings of the 36th International
  Workshop Graph Theoretic Concepts in Computer Science (WG)}, volume 6410 of
  {\em LNCS}, pages 266--278, 2010.

\bibitem{DBLP:conf/icalp/BonichonGHP10}
N.~Bonichon, C.~Gavoille, N.~Hanusse, and L.~Perkovic.
\newblock Plane spanners of maximum degree six.
\newblock In {\em Proceedings of the 37th International Colloquium, on
  Automata, Languages and Programming (ICALP)}, pages 19--30, 2010.

\bibitem{DBLP:journals/comgeo/BonichonGHP15}
N.~Bonichon, C.~Gavoille, N.~Hanusse, and L.~Perkovic.
\newblock Tight stretch factors for $L_1$- and $L_\infty$-{D}elaunay
  triangulations.
\newblock {\em Computational Geometry}, 48(3):237--250, 2015.

\bibitem{DBLP:journals/dcg/BonichonKPX15}
N.~Bonichon, I.~A. Kanj, L.~Perkovic, and G.~Xia.
\newblock There are plane spanners of degree 4 and moderate stretch factor.
\newblock {\em Discrete {\&} Computational Geometry}, 53(3):514--546, 2015.


\bibitem{DBLP:journals/ijcga/BoseDDOSSW12}  P. Bose,  M. Damian, K. Dou\"ieb,  J. O'Rourke,  B. Seamone,  M. Smid, and  S. Wuhrer.  $\pi/2$-Angle Yao Graphs are Spanners. {\em Int. J. Comput. Geom. Appl.}.   22(1): 61--82 (2012) 

\bibitem{el2009yao}N. El Molla. Yao spanners for wireless ad hoc networks. Master's thesis, Villanova University, 2009.

\bibitem{barba2013stretch} L. Barba,  P. Bose,  J. Carufel,  A.  Renssen, and   S. Verdonschot. On the stretch factor of the Theta-4 graph. {\em Workshop On Algorithms And Data Structures}. pp. 109--120, 2013.

\bibitem{DBLP:journals/comgeo/BoseMRV15} P. Bose,  P. Morin,  A. Renssen,  and  S. Verdonschot. The $\Theta_5$-graph is a spanner. {\em Comput. Geom.}.  48(2): 108--119,  2015.

\bibitem{DBLP:journals/dmaa/DamianR12} M. Damian, and   K. Raudonis. Yao Graphs Span Theta Graphs. {\em Discret. Math. Algorithms Appl.}. 4(2):1250014,  2012.

\bibitem{DBLP:journals/jocg/BarbaBDFKORTVX15} L. Barba,  P. Bose,  M. Damian,  R. Fagerberg,  W. Keng,  J. O'Rourke,  A. Renssen,  P. Taslakian,  S. Verdonschot, and G.  Xia. New and improved spanning ratios for Yao graphs. {\em J. Comput. Geom.}.  6(2): 19-53, 2015.



\bibitem{10.1007/978-3-030-83508-8_16}
P.~Bose, D.~Hill, and A.~Ooms.
\newblock Improved Bounds on the Spanning Ratio of the Theta-5-Graph.
\newblock In A.~Lubiw and M.~Salavatipour, editor, {\em Proceedings of the 17th Algorithm and Data Structures Symposium (WADS)}, pages 215--228. {Springer International Publishing}, 2021.

\bibitem{DBLP:conf/soda/BoseCHS19}
P.~Bose, J.~D. Carufel, D.~Hill, and M.~H.~M. Smid.
\newblock On the spanning and routing ratio of theta-four.
\newblock In T.~M. Chan, editor, {\em Proceedings of the Thirtieth Annual
  {ACM-SIAM} Symposium on Discrete Algorithms (SODA)}, pages 2361--2370.
  {SIAM}, 2019.

\bibitem{thetaUpper}
P.~Bose, J.~D. Carufel, P.~Morin, A.~van Renssen, and S.~Verdonschot.
\newblock Towards tight bounds on theta-graphs: More is not always better.
\newblock {\em Theoretical Computer Science}, 616:70--93, 2016.

\bibitem{thetaoriginalpaper1}
K.~Clarkson. 
\newblock Approximation algorithms for shortest path motion planning. 
\newblock In {\em Proceedings of the nineteenth annual ACM symposium on Theory of computing (STOC)},  pages 56–65 {Association for Computing Machinery}, 1987.

\bibitem{thetaoriginalpaper2}
J.~M~Keil. 
\newblock Approximating the complete Euclidean graph. 
\newblock In {\em Proceedings of the 1st Scandinavian workshop on algorithm theory (SWAT)},  pages 208–213 {Springer-Verlag}, 1988.

\bibitem{DBLP:journals/comgeo/BoseDLSV11}
P.~Bose, L.~Devroye, M.~L{\"{o}}ffler, J.~Snoeyink, and V.~Verma.
\newblock Almost all {D}elaunay triangulations have stretch factor greater than
  pi/2.
\newblock {\em Computational Geometry}, 44(2):121--127, 2011.

\bibitem{DBLP:journals/algorithmica/BoseGS05}
P.~Bose, J.~Gudmundsson, and M.~H.~M. Smid.
\newblock Constructing plane spanners of bounded degree and low weight.
\newblock {\em Algorithmica}, 42(3-4):249--264, 2005.

\bibitem{DBLP:journals/algorithmica/BoseHS18}
P.~Bose, D.~Hill, and M.~H.~M. Smid.
\newblock Improved spanning ratio for low degree plane spanners.
\newblock {\em Algorithmica}, 80(3):935--976, 2018.

\bibitem{DBLP:journals/comgeo/BoseS13}
P.~Bose and M.~H.~M. Smid.
\newblock On plane geometric spanners: {A} survey and open problems.
\newblock {\em Computational Geometry}, 46(7):818--830, 2013.

\bibitem{DBLP:journals/talg/ChengMV16}
S.~Cheng, L.~Mencel, and A.~Vigneron.
\newblock A faster algorithm for computing straight skeletons.
\newblock {\em ACM Transaction on Algorithms}, 12(3):44:1--44:21, 2016.

\bibitem{DBLP:journals/jcss/Chew89}
P.~Chew.
\newblock There are planar graphs almost as good as the complete graph.
\newblock {\em Journal of Computer and System Sciences}, 39(2):205--219, 1989.

\bibitem{DBLP:journals/jgaa/DehkordiFG15}
H.~R. Dehkordi, F.~Frati, and J.~Gudmundsson.
\newblock Increasing-chord graphs on point sets.
\newblock {\em Journal of Graph Algorithms and Applications}, 19(2):761--778, 2015.

\bibitem{DBLP:journals/dcg/DobkinFS90}
D.~P. Dobkin, S.~J. Friedman, and K.~J. Supowit.
\newblock Delaunay graphs are almost as good as complete graphs.
\newblock {\em Discrete {\&} Computational Geometry}, 5:399--407, 1990.

\bibitem{DBLP:journals/ijcga/DumitrescuG16}
A.~Dumitrescu and A.~Ghosh.
\newblock Lower bounds on the dilation of plane spanners.
\newblock {\em International Journal of Computational Geometry \& Applications}, 26(2):89--110, 2016.

\bibitem{DBLP:journals/dcg/EppsteinE99}
D.~Eppstein and J.~Erickson.
\newblock Raising roofs, crashing cycles, and playing pool: Applications of a
  data structure for finding pairwise interactions.
\newblock {\em Discrete {\&} Computational Geometry}, 22(4):569--592, 1999.

\bibitem{EppsteinGKT08}
D.~Eppstein, M.~T. Goodrich, E.~Kim, and R.~Tamstorf.
\newblock Motorcycle graphs: Canonical quad mesh partitioning.
\newblock {\em Computer Graphics Forum}, 27(5):1477--1486, 2008.

\bibitem{yaoold}
P.~Bose and A.~van Renssen.
\newblock Spanning Properties of Yao and $\Theta$-Graphs in the Presence of Constraints
\newblock {\em International Journal of Computational Geometry \& Applications}, 29(2):95--120, 2019.

\bibitem{fmmm}
S.~Hachul and M.~J\"{u}nger.
\newblock Large-graph layout with the fast multipole multilevel method.
\newblock {\em Technical Report. Cologne: University of Cologne, Computer
  Science Department}, 2005.

\bibitem{netX}
A.~A. Hagberg, D.~A. Schult, and P.~J. Swart.
\newblock Exploring network structure, dynamics, and function using networkx.
\newblock In {\em Proceedings of the 7th Python in Science Conference (SciPy)},
  2008.

\bibitem{Git}
B.~Hamedmohseni, D.~Mondal, and Z.~Rahmati.
\newblock Simplified emanation graph - implementations and tests.
\newblock \url{https://github.com/sneyes/SEG/tree/master}, 2019.
\newblock Online; accessed 6 June 2019.

\bibitem{DBLP:conf/cccg/HamedmohseniRM18}
B.~Hamedmohseni, Z.~Rahmati, and D.~Mondal.
\newblock Emanation graph: {A} new $t$-spanner.
\newblock In S.~Durocher and S.~Kamali, editors, {\em Proceedings of the 30th
  Canadian Conference on Computational Geometry (CCCG)}, pages 311--317, 2018.

\bibitem{DBLP:conf/sofsem/HamedmohseniRM20}
B.~Hamedmohseni, Z.~Rahmati, and D.~Mondal.
\newblock Simplified emanation graphs: {A} sparse plane spanner with steiner
  points.
\newblock In A.~Chatzigeorgiou, R.~Dondi, H.~Herodotou, C.~A. Kapoutsis,
  Y.~Manolopoulos, G.~A. Papadopoulos, and F.~Sikora, editors, {\em Proceedings
  of the 46th International Conference on Current Trends in Theory and Practice
  of Computer Science (SOFSEM)}, volume 12011 of {\em LNCS}, pages 607--616.
  Springer, 2020.

\bibitem{DBLP:journals/jocg/KanjPT17}
I.~A. Kanj, L.~Perkovic, and D.~T{\"{u}}rkoglu.
\newblock Degree four plane spanners: Simpler and better.
\newblock {\em Journal of Computational Geometry (JoCG)}, 8(2):3--31, 2017.

\bibitem{DBLP:journals/dcg/KeilG92}
J.~M. Keil and C.~A. Gutwin.
\newblock Classes of graphs which approximate the complete {E}uclidean graph.
\newblock {\em Discrete {\&} Computational Geometry}, 7:13--28, 1992.

\bibitem{abs-1801-06290}
A.~Lubiw and D.~Mondal.
\newblock Angle-monotone graphs: Construction and local routing.
\newblock {\em CoRR}, abs/1801.06290, 2018.

\bibitem{mondal2017}
D.~Mondal and L.~Nachmanson.
\newblock A new approach to {GraphMaps}, a system browsing large graphs as
  interactive maps.
\newblock In {\em Proceedings of the 13th International Joint Conference on
  Computer Vision, Imaging and Computer Graphics Theory and Applications
  (VISIGRAPP)}, pages 108--119, 2018.

\bibitem{DBLP:conf/gd/NachmansonPLRHC15}
L.~Nachmanson, R.~Prutkin, B.~Lee, N.~H. Riche, A.~E. Holroyd, and X.~Chen.
\newblock {GraphMaps}: Browsing large graphs as interactive maps.
\newblock In {\em Proceedings of the 23rd International Symposium on Graph
  Drawing and Network Visualization (GD)}, pages 3--15, 2015.

\bibitem{DBLP:conf/imr/Owen98}
S.~J. Owen.
\newblock A survey of unstructured mesh generation technology.
\newblock In {\em Proceedings of the 7th International Meshing Roundtable
  (IMR)}, pages 239--267, 1998.

\bibitem{shewchuk96b}
J.~R. Shewchuk.
\newblock Triangle: Engineering a {2D} quality mesh generator and {D}elaunay
  triangulator.
\newblock In M.~C. Lin and D.~Manocha, editors, {\em Proceedings of the Applied
  Computational Geormetry, Towards Geometric Engineering (FCRC)}, volume 1148
  of {\em LNCS}, pages 203--222. Springer, 1996.

\bibitem{Triangle}
Triangle.
\newblock A two-dimensional quality mesh generator and {D}elaunay triangulator.
\newblock \url{https://www.cs.cmu.edu/~quake/triangle.html}, 2013.
\newblock Online; accessed 6 June 2019.

\bibitem{delauUpper}
G.~Xia.
\newblock The stretch factor of the {D}elaunay triangulation is less than
  1.998.
\newblock {\em SIAM Journal on Computing}, 42(4):1620--1659, 2013.

\bibitem{delauLower}
G.~Xia and L.~Zhang.
\newblock Toward the tight bound of the stretch factor of {D}elaunay
  triangulations.
\newblock In {\em Proceedings of the 23rd Annual Canadian Conference on
  Computational Geometry (CCCG)}, pages 175--180, 2011.

\bibitem{DBLP:journals/siamcomp/Yao82}
A.~C. Yao.
\newblock On constructing minimum spanning trees in $k$-dimensional spaces and
  related problems.
\newblock {\em SIAM Journal on Computing}, 11(4):721--736, 1982.
\end{thebibliography}


\end{document}